\documentclass[runningheads]{llncs}
 \usepackage[utf8]{inputenc}
\usepackage{latexsym,amsfonts}

\usepackage{tikz}
\usetikzlibrary{calc,shapes,arrows,automata}
\usepackage{amsmath}
\usepackage{algorithm}
\usepackage{subcaption} 
\usepackage{caption}
\usepackage{bm}
\usepackage{lineno}
\usepackage{listings}
\usepackage{algorithm}
\usepackage{algorithmicx}
\usepackage{algpseudocode}
\usepackage{epsfig}
\usepackage[hyphens]{url} 
\usepackage[misc,geometry]{ifsym}

\newcommand{\R}{\mathbb{R}}
\newcommand{\N}{\mathbb{N}}

\newcommand{\Bb}{\mathcal{B}}

\newcommand{\Xx}{\mathcal{X}}

\newcommand{\Sys}{\mathfrak{S}}
\newcommand{\Ll}{\mathcal{L}}

\newcommand{\Inf}{\mathsf{Inf}}
\newcommand{\Tt}{\mathcal{T}}
\newcommand{\Vv}{\mathcal{V}}

\newcommand{\Aa}{\mathcal{A}}

\newcommand{\Tr}{\mathsf{Tr}_{(\Sys, \Ll, \kappa)}}

\newcommand{\Reach}{\mathsf{Reach}}

\newcommand{\set}[1]{\{ #1 \}}

\newcommand{\Zz}{\mathcal{Z}}
\begin{document}

\title{Control Closure Certificates} 

\author{Vishnu Murali (\Letter), Mohammed Adib Oumer, and Majid Zamani \thanks{This work was supported by the NSF grants CNS-2111688 and CNS-2145184.} }
\authorrunning{V. Murali, M.A. Oumer and M. Zamani}
\institute{University of Colorado Boulder}
\maketitle

\begin{abstract}
    \label{sec:abstract}
    This paper introduces the notion of control closure certificates ($\text{C}^3$)  to synthesize controllers for discrete-time control systems against  $\omega$-regular specifications.
    Typical functional approaches to synthesize controllers against $\omega$-regular specifications rely on combining inductive invariants (for example, via barrier certificates) with proofs of well-foundedness (for example, via ranking functions).
    Transition invariants, provide an alternative where instead of standard well-foundedness arguments one may instead search for disjunctive well-foundedness arguments that together ensure a well-foundedness argument. 
     Closure certificates, functional analogs of transition invariants, provide an effective, automated approach to verify discrete-time dynamical systems against linear temporal logic and $\omega$-regular specifications.
     We build on this notion to synthesize controllers to ensure the satisfaction of  $\omega$-regular specifications.
     To do so, we first illustrate how one may construct control closure certificates to visit a region infinitely often (or only finitely often) via disjunctive well-founded arguments.
     We then combine these arguments to provide an argument for parity specifications.
     Thus, finding an appropriate $\text{C}^3$ over the product of the system and a parity automaton specifying a desired $\omega$-regular specification ensures that there exists a controller $\kappa$ to enforce the $\omega$-regular specification.
     We propose a sum-of-squares optimization approach to synthesize such certificates and demonstrate their efficacy in designing controllers over some case studies.
\end{abstract}


\section{Introduction} 
\label{sec:intro}

We introduce a notion of control closure certificates to synthesize controllers for discrete-time control systems against $\omega$-regular specifications.
Closure certificates~\cite{murali2024closure}, are a functional analog of transition invariants to verify discrete-time systems against $\omega$-regular specifications.
Unlike barrier certificates \cite{prajna_2004_safety}, which seek to overapproximate the reachable states of a system, closure certificates are real-valued functions that seek to overapproximate the reachable \emph{transitions} of a system.
When used for safety specifications, such certificates often allow for much simpler templates \cite[Section 3.1]{murali2024closure} to prove safety when compared to traditional barrier certificates.
When combined with proofs of well-foundedness \cite{cook_2009_priciples}, such certificates may be leveraged to prove $\omega$-regular specifications.
Similar to barrier certificates, one may automate the search for closure certificates via optimization techniques such as sum-of-squares (SOS) \cite{Parrilo_2003}, satisfiability modulo theory solvers (SMT) \cite{moura_2011_smt,moura_2008_z3,gao_dreal_2013}, and neural networks \cite{abate2021fossil,nadali_2024_closure}.
In the above cases, one fixes the certificate to be within a template class, \textit{e.g.}, polynomials of a fixed degree or neural networks of a fixed size, and then proceeds to search for such a certificate within the template class.
We build on this work to show how one may simultaneously search for both a certificate, as well as a controller to ensure the satisfaction of a desired $\omega$-regular specification.
To do so we first provide certificate conditions to ensure that a set is visited infinitely often, or that it is visited only finitely often. 
We then show how one may combine these proof techniques to design controllers for more general $\omega$-regular specifications.

\noindent \textbf{Functional proofs for $\omega$-regular specifications.}
Functional approaches to synthesize controllers for $\omega$-regular specifications rely on a combination of finding inductive invariants with well-foundedness arguments \cite{dimitrova_2014_deductive,abate2024stochastic,chatterjee2024sound}.
Such a proofs of well-foundedness, similar to proofs of termination \cite{cook_2009_priciples}, can be reduced to finding an appropriate ranking function $\Vv$. 
To illustrate these approaches, let $\Xx$ denote the set of states of the system, and assume that one is able to determine the exact reachable set of a system (denoted $\Reach$), and let $\Xx_{VF}$ be a set of states of a system that must be visited only finitely often.
Consider a function $\Vv: \Xx \to \R$ from the states of the system to the reals such that:
\begin{align}
    & \Vv(x') \leq \Vv(x) - \xi, &&  \text{ if }x \in \Xx_{VF} \cap \Reach \text{ and }  \nonumber \\
    & \Vv(x') \leq \Vv(x) && \text{ if } x \in \Reach \setminus \Xx_{VF}, \nonumber
\end{align}
where $x'$ is the one-step transition from state $x$ and $\xi$ is some positive value. 
Then the existence of a function $\Vv$ that is bounded from below and satisfies the above conditions provides a proof that the system visits the set $\Xx_{VF}$ only finitely often. 
As one typically does not know the set $\Reach$, one instead tries to ensure that the above conditions hold over an inductive state invariant that overapproximates the set $\Reach$.
The promise of such functional proofs lie in their automatability.
One may effectively search for them via optimization \cite{prajna_2004_safety,Parrilo_2003} or learning-based \cite{abate2021fossil} techniques.

\noindent \textbf{Transition invariants.}
In the above discussion we considered the existence of a single function $\Vv$ to prove a set $\Xx_{VF}$ is visited only finitely often.
Practically, one may want to instead consider partitions over relevant sets and try to find independent ranking function arguments to prove this.
Unfortunately, such a strategy fails to be sound in general.
This insight was observed in \cite{podelski_2004_transition}, which showed that one cannot prove a relation to be well-founded if it is a union of well-founded relations. 
However, it is possible to prove a relation is well-founded if its transitive closure (cf. Section \ref{sec:prelims}) is disjunctively well-founded.
To overapproximate the transitive closure, they introduced a notion of \emph{transition invariants}, that overapproximate the reachable transitions of a system.
Hence one is able to leverage transition invariants to find independent ranking functions to prove well-foundedness.
While transition invariants provide an approach to prove a set is visited only finitely often, one often requires to show the dual, that a set be visited infinitely often, or both.
This is the case when dealing with deterministic automata that describe $\omega$-regular languages such as Rabin, Streett, or parity automata.
Thus, we consider the question of how disjunctive well-founded proofs help in designing controllers for $\omega$-regular specifications. 

\noindent \textbf{Our Contributions.}
\begin{itemize}
    \item We introduce a notion of control closure certificates to synthesize controllers for discrete-time control systems against $\omega$-regular specifications.
    \item We show how proofs of disjunctive well-foundedness are useful not just for showing a set is visited only finitely often but also in proofs of showing a set is visited infinitely often, thus making them amenable to provide conditions for synthesis against $\omega$-regular properties. 
    \item We rely on an optimization-based approach to automate the search for these control closure certificates and demonstrate their use in some case studies.
\end{itemize}

\noindent \textbf{Related works.} 
The results in \cite{prajna_2004_safety} introduced a notion of barrier certificates to act as functional inductive invariants for hybrid systems.
These results illustrated how one may effectively automate the search for such  such certificates via optimization techniques such as sum-of-squares programming \cite{Parrilo_2003}.
Building on this idea, the results in \cite{dimitrova_2014_deductive} considered a notion of parity certificates to synthesize controllers for objectives specified by parity automata and then demonstrated how one may use such certificates in verifying alternating-time temporal logic properties \cite{alur2002alternating} .
These certificates relied on a combination of invariant arguments (characterized as barrier certificates) with ranking functions (described as Lyapunov functions) to design controllers.
Similar to these results, the results in \cite{abate2024stochastic} proposed a notion of Streett supermartingales to synthesize controllers for stochastic systems against Streett objectives.
The results in \cite{chatterjee2024sound} introduced a notion of B\"uchi ranking functions to provide a sound and semi-complete proof rule for real-valued programs via Putinar's Positivestellensatz \cite{putinar1993positive}.
More recent works consider extensions for stochastic systems \cite{majumdar2024necessary,abate2025quantitative,henzinger2025supermartingale}.
The results in \cite{murali2024closure} proposed a notion of closure certificates to describe functional transition invariants analogous to barrier certificates and considered techniques to automate the search for these certificates via sum-of-squares (SOS) programming and satisfiability modulo theory (SMT) \cite{moura_2011_smt} solvers.
Building on this, the results in \cite{wang2025verifying} considered their use for recurrence (showing a set is visited infinitely often).
Our work builds on the notion of closure certificates to find transition invariants and proofs of disjunctive well-foundedness to synthesize controllers for specifications characterized by parity automata.
We show that while one may directly adapt conditions for recurrence as in \cite{wang2025verifying}, one faces challenges when trying to prove disjunctive well-foundedness (cf. Section \ref{subsec:ccc_buchi}).
We thus propose alternative conditions to show a set is visited infinitely often by relying on proofs of disjunctive well-foundedness.
These disjunctive well-foundedness proofs may be used as acceleration lemmas \cite{heim2024solving}, similar to the ones used for ranking functions and state invariants.
Another approach to design controllers for $\omega$-regular specification is to construct a finite abstractions of the system and design controllers for these finite abstractions as illustrated by \cite{henzinger_1997_hytech,lahijanian_2011_temporal,rungger_2016_scots,khaled_2021_omegathreads}.
These build on existing synthesis techniques for finite-state systems \cite{baier_2008_principles,tabuada_2009_verification}.

\section{Preliminaries}
\label{sec:prelims}
We denote the set of natural numbers and reals by $\N$ and $\R$ respectively.
Given a number $a \in \R$, we use $\R_{\geq a}$ and $\R_{> a}$ to denote the intervals $[a, \infty[$ and $]a,\infty[$ respectively, and similarly, for any natural number $n \in \N$, we use $\N_{\geq n}$ to denote the set of natural numbers greater than or equal to $n$.
Let $R \subseteq A \times B$ be a relation, and $a \in A$, we use $R(a)$ to denote the set $\{ b \mid (a,b) \in R \}$. 
Given two sets $A,B$, we use $A \setminus B$ to denote the set containing those elements that are present in $A$ but not in $B$, and as usual use $A \cup B$ and $A \cap B$ to represent their union and intersection.
Given a set $R \subseteq A \times A$ and any $i \in \N_{\geq 1}$,  we define $R^i$ recursively as $R^{1} = \{ (a_1, a_2) \mid (a_1, a_2) \in R  \}$ if $i = 1$, and $R^{i} = \{ (a_1, a_2) \mid (a_1, a_3) \in R^{i-1}, \text{ and } (a_3, a_2) \in R   \}$.
That is $R^i$ is the $i$-fold self-composition of the relation with itself.
Given a relation $R \subseteq A \times A$, its transitive closure is defined as the set $R^{+} := \underset{i \in \N_{ \geq 1}}{\bigcup} R^i$. 
Finally, we use logical operators $\wedge$, and $\implies$ as shorthands for conjunction, and implication, respectively.
Given a pair $(a,b) \in A \times B$, we use $\pi_1(a,b) = a$, and $\pi_2(a,b) = b$ to denote projections of these pairs.
We say that a function $f: A \to \R$ is bounded from below (resp. bounded from above) if there exists some $l \in \R$, such that $ l \leq f(a)$ for all $a \in A$ (resp. if there exists some $u \in \R$ such that $f(a) \leq u$ for all $a \in A$).
A function is bounded if it is bounded from below and above.
Similarly, we say a function $f: A \times B \to \R$ is bounded from below (resp. above) in $A$ if, for all $b \in B$ there exist $l \in \R$ (resp. $u \in \R$) such that $l \leq f(a,b)$ ($f(a,b) \leq u$) for all
$a \in A$.
Given a set $A$, denote the set of finite and countably infinite sequences of elements in $A$ by sets $A^{*}$ and $A^{\omega}$ respectively. 
We use the notation $\bm{a} = \langle a_1, a_2, \ldots, a_n \rangle \in A^{*}$ for finite length sequences and $\bm{s} = \langle a_0, a_1, \ldots \rangle\in A^{\omega}$ for infinite-sequences.
Given a sequence $s = \langle a_0, a_1, \ldots \rangle$ we say that the sequence $\langle b_0, b_1, \ldots \rangle$ is a subsequence of $s$ iff $b_i = a_{j_i}$ where $j_0 \leq j_1 \leq \ldots$ for all $i \in \N$, 
Let $\Inf(\bm{s})$ be the set of infinitely often occurring elements in the sequence $\bm{s}  =\langle a_0, a_1, \ldots \rangle$.
Given a possibly infinite sequence $\bm{s} = \langle a_0, a_1, \ldots \rangle$, and two natural numbers $i,j \in \N$ where $i \leq j$, we use $\bm{s}[i,j]$ to indicate the finite sequence $\langle a_i, a_{i+1}, \ldots, a_j \rangle$, and $\bm{s}[i, \infty[$ to indicate the infinite sequence $\langle a_i, a_{i+1}, \ldots \rangle$. 
Finally, we use $\bm{s}[i]$ to denote the $i$th element in the sequence $\bm{s}$ for any $i \in \N$. 
With a slight abuse of notation, following \cite{podelski_2004_transition}, we say that a relation $R \subseteq A \times A$ is well-founded if nthere exists no infinite sequence $\langle a_0, a_1, \ldots \rangle$ such that $(a_i, a_{i+1}) \in R$ for all $i \in \N$.
In this work, we consider a relation $R \subseteq A \times A$ to be well-founded with respect to a set $B \subseteq A$, if for every  infinite sequence $\langle a_0, \ldots \rangle \in A^{\omega}$, where $(a_i, a_{i+1}) \in R$, there does not exist a subsequence $\langle b_0, b_1, \ldots \rangle \in B^{\omega}$.
A relation $R \subseteq A \times A$ is said to be disjunctively well-founded if it is the finite union of well-founded relations, \textit{i.e.}, $R = \bigcup_{i = 1}^{m} R_i$ for some $m \in \N$, where the relations $R_i$ are well-founded. 

\subsection{Discrete-time Control Systems}
\label{subsec:prelims_system}
A discrete-time control system (simply, a control system) $\Sys$ is a tuple $(\Xx,\Xx_0, U, f)$, where $\Xx \subseteq \R^n$ denotes the state set, $\Xx_0 \subseteq X$ denotes a set of initial states, $U \subseteq \R^m$ denotes the set of control inputs, and $f:\Xx \times U \to \Xx$ the state transition function.
We assume that state set of the systems under consideration are compact.
Given a control sequence $\bm{u} = \langle u_0, u_1, \ldots \rangle \in U^{\omega}$, and an initial state $x_0 \in \Xx_0$, the corresponding  \emph{state sequence} is the infinite sequence  $\bm{x}_{\bm{u}} = \langle x_0, x_1, \ldots \rangle \in \Xx^{\omega}$ where and $x_{i+1} = f(x_i, u_i)$, for all $i \in \N$.
In this paper, we consider the systems to be controlled with state feedback controllers, \textit{i.e.}, controllers of the form $\kappa: \Xx \to 2^U$, where for all states $x \in \Xx$ one may select a choice of input \footnote{We should add that in the case of $\omega$-regular specifications, we consider controllers that are state feedback over the product of the system and the desired automaton, as well as controllers with a finite amount of memory.} $u \in \kappa(x)$ or controllers with a finite memory, where $\kappa: \Xx \times \{0, \ldots mem \} \to 2^U$ (cf. Definitions \ref{def:ccc_inf_vis_counter} and \ref{def:ccc_inf_vis_counter_disj}, where we use a counter $j$ as memory). 
Thus, given an initial state $x_0 \in \Xx_0$, a state trajectory of the system under controller $\kappa$ is the infinite sequence $\langle x_0, x_1, \ldots \rangle$
such that $x_{i+1} = f(x_i, u_i)$, where $u_i \in \kappa(x_i)$ for all $i \in \N$.
For a finite alphabet set $\Sigma$, we associate a labeling function $\Ll: \Xx \to \Sigma$ which maps each state of the system to a letter in $\Sigma$. 
This naturally generalizes to mapping a state sequence of the system $\langle x_0, x_1, \ldots \rangle \in \Xx^{\omega}$ to a trace or word $w = \langle \Ll(x_0), \Ll(x_1), \ldots \rangle \in \Sigma^{\omega}$. 
Finally, let $\Tr$ denote the set of all traces of system $\Sys$ under the labeling map $\Ll$ and controller $\kappa$. 
For convenience, when $U$ is singleton, we use $\Sys_{dyn} = (\Xx, \Xx_0, f)$, to denote a dynamical system with constant (or no) input, \textit{i.e.}, $f: \Xx \to \Xx$ is the state transition function.
To motivate the use of closure certificates to synthesize controllers for $\omega$-regular specifications, we first draw an analogy to the use of control barrier certificates to ensure safety. 

\subsection{Safety Verification and Barrier Certificates}
\label{subsec:prelims_safety}
A control system $\Sys = (\Xx, \Xx_0, U, f)$ is safe with respect to a set of unsafe states $\Xx_u$, if no state sequence reaches $\Xx_u$, \textit{i.e.}, for every state sequence  $\langle x_0, x_1, \ldots  \rangle$, we have $x_i \notin \Xx_u$ for all $i \in \N$.
\begin{definition}
\label{def:bar}
    A function $\Bb: \Xx \to \R$ is a control barrier certificate for a system $\Sys$ with respect to a set of unsafe states $\Xx_u$ if: 
    \begin{align}
        & \Bb(x) \leq 0, && \text{ for all } x \in \Xx_0, \label{eq:bar_cond_1} \\
        & \Bb(x) > 0, && \text{ for all } x \in \Xx_u ,\label{eq:bar_cond_2} 
    \end{align}
    and for all $x \in \Xx$, there exists $u \in U$ such that:
    \begin{align}
        & \big( \Bb(x) \leq 0 \big) \implies \big( \Bb(f(x,u)) \leq 0 \big) && \label{eq:bar_cond_3}
    \end{align}
\end{definition}
\begin{theorem}
    Consider a control system $\Sys = (\Xx, \Xx_0, U, f)$, with a set of unsafe states $\Xx_u$. The existence of a function $\Bb:\Xx \to \R$ satisfying conditions~\eqref{eq:bar_cond_1}-\eqref{eq:bar_cond_3} implies that there exists a controller $\kappa$ to ensure that the system is safe.
\end{theorem}

\subsection{B\"uchi and Parity automata}
\label{subsec:prelims_aut}
We now discuss some classes of $\omega$-regular automata to capture our specifications of interest.
An $\omega$-regular automaton $\Aa$ is a tuple $(\Sigma,Q, Q_0, \delta,Acc)$, where 
$\Sigma$ denotes a finite alphabet, $Q$ a finite set of states, $Q_0 \subseteq Q$ an initial set of states, $\delta \subseteq Q \times \Sigma \times Q$ the transition relation, and $Acc$ denotes its accepting condition.
If the accepting condition is B\"uchi, then we call that automaton a nondeterministic B\"uchi automaton (NBA), and we have $Acc \subseteq Q$.
An automaton is a nondeterminstic parity automaton (NPA) if the accepting condition $Acc: Q \to \{1, \ldots, c\}$ maps each state $q \in Q$ to some color (denoted by a natural number).
A run of the automaton $\Aa = (\Sigma,Q, q_0, \delta, Acc)$ over a word $w = \langle \sigma_0, \sigma_1, \sigma_2 \ldots \rangle \in \Sigma^{\omega}$, is an infinite sequence of states $\rho = \langle q_0,q_1, q_2, \ldots \rangle \in Q^{\omega}$ with $q_0 \in Q_0$ and $q_{i+1} \in \delta(q_i, \sigma_i)$.
An NBA $\Aa = (\Sigma,Q, Q_0, \delta, Acc)$ is said to accept a word $w$, if there exists a run $\rho$ on $w$ where $\Inf(\rho) \cap Acc \neq \emptyset$.
An NPA $\Aa = (\Sigma,Q, Q_0, \delta, Acc)$ is said to accept a word $w$, if there exists a run $\rho$ on $w$ if we have the minimum priority seen infinitely often in $\rho$ is even (equivalently one may consider maximum priorities, or colors that are odd).
We denote the set of words accepted by an automaton $\Aa$ (the language of the automaton) as $L(\Aa)$.
Finally, we say that an automaton is \emph{deterministic} if $|Q_0| = 1$ and for all $q \in Q$, and $\sigma \in \Sigma$, we have $| \delta(q, \sigma)| \leq 1$.
We use  DBA or DPA to denote determinstic B\"uchi or determinstic parity automata, respectively.
As $|Q_0| = 1$, we use $q_0$ to denote the initial state for a DPA or DBA in the tuple, \textit{i.e.}, a DPA is of the form $\Aa = (\Sigma, Q, q_0, \delta, Acc)$.
Note that both NBAs and DPAs are closed under complementation~\cite{safra_1988_complexity}: given an NBA (resp. DPA) $\Aa = (\Sigma,Q, Q_0, \delta, Acc)$, there exists an NBA (resp. DPA) $\Aa' = (\Sigma, Q', Q'_0, \delta', Acc)$ such that $L(\Aa') = \overline{L(\Aa)}$.

\subsection{Problem Statement}
\label{sec:statement}
The key problem we consider is as follows: Given a control system $\Sys = (\Xx, \Xx_0, U ,f)$, a deterministic parity automaton $\Aa = (\Sigma, Q, Q_0, \delta, Acc)$, and a labeling map $\Ll: \Xx \to \Sigma$, find a controller $\kappa: \Xx \times Q \to 2^U$ such that $\Tr \subseteq L(\Aa)$.
To tackle this problem we introduce a notion of control closure certificates.

\section{Control Closure Certificates (C$^3$s)}
\label{sec:ccc}
Typical inductive state invariants seek to overapproximate the reachable states of a system. 
Transition invariants \cite{podelski_2004_transition}, on the other hand, seek to to overapproximate the reachable transitions of a system.
Such invariants can be inductively characterized as follows.
Given a dynamical system $\Sys_{dyn} = (\Xx, \Xx_0, f)$, a relation $R \subseteq \Xx \times \Xx$ is a transition invariant if:
\begin{align}
    & (x, f(x)) \in R && \text{ for all } x \in \Xx, \text{ and } \nonumber \\
    & ((f(x), y) \in R) \implies (x, y) \in R && \text{ for all } x, y \in \Xx. \nonumber
\end{align}
Building on this intuition, the results in \cite{murali2024closure} considered a notion of closure certificates that act as functional transition invariants.
It was demonstrated that such functions can be used to verify $\omega$-regular properties as well as provide simpler templates of functions compared to existing approaches.
In the following sections, we describe how closure certificates can be used to design controllers to ensure a set is visited either  only finitely often or infinitely often. 
We then demonstrate how one may combine these conditions to design controllers to satisfy objectives specified by parity automata.
We should add that as parity, Rabin, and Streett are equally expressive, one can effectively consider alternate conditions for Rabin or Streett automata.
We omit these conditions but note that one might combine certificates in a similar fashion as parity automata.
We first start with designing $\text{C}^3$s to synthesize controllers that ensure a given set is visited only finitely often. 
 
\subsection{C$^3$s for Finite Visits}
\label{subsec:ccc_cobuchi}
In this section, our objective is to design a controller to ensure that a system visits a set $\Xx_{VF} \subseteq \Xx$ only finitely often via $\text{C}^3$s.
First, we discuss how one may leverage disjunctive well-foundedness to verify such a condition as follows.
\begin{definition}
    \label{def:tbar_fin_vis_verif}
    Consider a dynamical system $\Sys_{dyn} = (\Xx, \Xx_0, f)$ and a set of states $\Xx_{VF} \subseteq \Xx$ that must be visited only finitely often.
    Let $\Xx_{VF}$ be partitioned into sets $\Xx_{VF_{1}}, \ldots, \Xx_{VF_{p}}$, \textit{i.e.}, $\Xx_{VF} = \underset{0 \leq i \leq p}{\bigcup} \Xx_{VF_i}$ for some $p \in \N$. 
    Then, function $\Tt: \Xx \times \Xx \to \R$, and bounded (from below) functions $\Vv_i: \Xx \to \R_{\geq 0}$ for all $0 \leq i \leq p$ are a disjunctive closure certificate if for all $x,y \in \Xx$:
    \begin{align}
        &\big( \Tt(x, f(x)) \geq 0 \big), \label{eq:tbar_cond_1_fin_vis_verif} \\
        & \big( \Tt(f(x), y) \geq 0 \big) \implies \big( \Tt(x, y) \geq 0 \big), \label{eq:tbar_cond_2_fin_vis_verif}
    \end{align}
    and for all $x_0 \in \Xx_0$, and any $0 \leq i \leq p$, there exists $\xi_i \in \R_{ > 0}$ such that for all $z,z' \in \Xx_{VF_{i}}$, we have:
    \begin{align}
            &\big( \Tt(x_0, z) \geq 0 \big) \wedge \big( \Tt(z,z') \geq 0 \big) \implies \big( \Vv_i(x_0, z') \leq \Vv_i(x_0, z) - \xi_i \big).   \label{eq:tbar_cond_3_fin_vis_verif}    
    \end{align}
\end{definition}
\begin{lemma}
    \label{lem:cc_fin_vis}
    Consider a dynamical system $\Sys_{dyn} = (\Xx, \Xx_0, f)$ and a set $\Xx_{VF}$ that must be visited only finitely often.
    The existence of functions $\Tt$ and $\Vv$ as in Definition \ref{def:ccc_fin_vis} guarantees that  $\Xx_{VF}$ is visited only finitely often.
\end{lemma}

The proof of Lemma \ref{lem:cc_fin_vis} can be found in Appendix \ref{ap:pr_lem_cc_fin}.

A direct approach to consider a notion of control closure certificates is to add an existentially quantified control input to conditions \eqref{eq:tbar_cond_1_fin_vis_verif} and \eqref{eq:tbar_cond_2_fin_vis_verif} as follows.

\begin{definition}
    \label{def:ccc_fin_vis_alt}
    Consider a control system $\Sys = (\Xx, \Xx_0, U, f)$ and a set of states $\Xx_{VF}$ that must be visited only finitely often.
    Let $\Xx_{VF}$ be partitioned into sets $\Xx_{VF_{1}}, \ldots, \Xx_{VF_{p}}$, \textit{i.e.}, $\Xx_{VF} = \underset{0 \leq i \leq p}{\bigcup} \Xx_{VF_i}$ for some $p \in \N$. 
    Then, function $\Tt: \Xx \times \Xx \to \R$, and bounded (from below) functions $\Vv_i: \Xx \times \Xx \to \R_{\geq 0}$ for all $0 \leq i \leq p$ are a control closure certificate if for all $x \in \Xx$ there exists $u \in U$ such that for all $y \in Y$ we have:
    \begin{align}
        &\big( \Tt(x, f(x,u)) \geq 0 \big), \label{eq:tbar_cond_1_fin_vis_alt} \\
        & \big( \Tt(f(x,u), y) \geq 0 \big) \implies \big( \Tt(x, y) \geq 0 \big), \label{eq:tbar_cond_2_fin_vis_alt}
    \end{align}
    and for all $x_0 \in \Xx_0$, $0 \leq i \leq p$, and for all $z,z' \in \Xx_{VF_{i}}$, there exists $\xi_i \in \R_{ > 0}$ such that:
    \begin{align}
            &\big( \Tt(x_0, z) \geq 0 \big) \wedge \big( \Tt(z,z') \geq 0 \big) \implies
        \big( \Vv_i(x_0, z') \leq \Vv_i(x_0, z) - \xi_i \big).   \label{eq:tbar_cond_3_fin_vis_alt}         
    \end{align}
\end{definition}

Observe that the conditions above rely on two quantifier alternations over the state set and input set (between $x$ and $u$, and $u$ and $y$) rather than one.
We now show how one can avoid this alternation by considering an alternative paradigm where we define C$^3$s as follows.

\begin{definition}
    \label{def:ccc_fin_vis}
    Consider a control system $\Sys = (\Xx, \Xx_0, U, f)$ and a set of states $\Xx_{VF}$ that must be visited only finitely often.
    Let $\Xx_{VF}$ be partitioned into sets $\Xx_{VF_{1}}, \ldots, \Xx_{VF_{p}}$, \textit{i.e.}, $\Xx_{VF} = \underset{0 \leq i \leq p}{\bigcup} \Xx_{VF_i}$ for some $p \in \N$. 
    Then, function $\Tt: \Xx \times \Xx \to \R$, and bounded (from below) functions $\Vv_i: \Xx \to \R_{\geq 0}$ for all $0 \leq i \leq p$ are a C$^3$ if for all $x \in \Xx$ there exists $u \in U$ such that:
    \begin{align}
        &\big( \Tt(x, f(x,u)) \geq 0 \big), \label{eq:tbar_cond_1_fin_vis}
    \end{align}
    and for all $x, y \in \Xx$,  for all $u \in U$, we have:
    \begin{align}
        & \big(  \Tt(x, f(x,u) \big) \geq 0 \big) \implies \Big( \big( \Tt(f(x,u), y) \geq 0 \big) \implies \big( \Tt(x, y) \geq 0 \big) \Big), \label{eq:tbar_cond_2_fin_vis}
    \end{align}
    and for all $x_0 \in \Xx_0$, $0 \leq i \leq p$, and for all $z,z' \in \Xx_{VF_{i}}$, there exists $\xi_i \in \R_{ > 0}$ such that:
    \begin{align}
            &\big( \Tt(x_0, z) \geq 0 \big) \wedge \big( \Tt(z,z') \geq 0 \big) \implies 
         \big( \Vv_i(x_0,z') \leq \Vv_i(x_0,z) - \xi_i \big).   \label{eq:tbar_cond_3_fin_vis}         
    \end{align}
\end{definition}

\begin{theorem}
    \label{thm:ccc_fin_vis}
    Consider a control system $\Sys = (\Xx, \Xx_0, U, f)$ and a set $\Xx_{VF}$ that must be visited only finitely often.
    The existence of functions $\Tt$ and $\Vv$ as in Definition \ref{def:ccc_fin_vis} guarantees that there exists a controller $\kappa$ to ensure that $\Xx_{VF}$ is visited only finitely often.
\end{theorem}

\begin{proof}
    We prove Theorem \ref{thm:ccc_fin_vis} via contradiction.
    To do so, assume that there exists an initial state $x_0 \in \Xx_0$, such that for all input sequences $\bm{u}  = \langle u_0, u_1, \ldots \rangle $, we have the corresponding state sequence $\bm{x}_{\bm{u}} = \langle x_0, x_1, \ldots \rangle$ to visit the set $\Xx_{VF}$ infinitely often.
    Consider the control input sequence selected such that $\Tt(x_i, f(x_i, u_i) ) \geq 0$ for all $i \in \N$.
    As condition \eqref{eq:tbar_cond_1_fin_vis} holds, this is true for any $x_i$ in the state sequence.
    Following conditions \eqref{eq:tbar_cond_2_fin_vis} and \eqref{eq:tbar_cond_1_fin_vis} and via induction, we have $
\Tt(x_0, x_i) \geq 0$ and $\Tt(x_i, x_{j}) \geq 0$ for all $i \in \N$, and all $j \geq (i+1)$.
Let $\langle y_0, y_1, \ldots \rangle$ be the subsequence that visits $\Xx_{VF}$ only finitely often.
That is the state sequence is of the form $\bm{x}_{\bm{u}} = \langle x_0, x_1, \ldots, y_0, \ldots, y_1, \ldots \rangle$.
Via Ramsey's theorem \cite{ramsey1987problem}, there exists a subsequence $\langle z_0, z_1, \ldots \rangle \in \Xx_{VF_i}$ that visits $\Xx_{VF_i}$ infinitely often for some $0 \leq i \leq p$.
From the previous results, we know that $\Tt(x_0, z_i) \geq 0$ and $\Tt(z_i, z_j) \geq 0$ for all $i \in \N$, and all $j \geq (i+1)$.
Let $ \Vv_i^* := \Vv_i(x_0,z_0) $ and as function $\Vv_i$ is bounded from below let the lower bound be $\Vv_i^{\dagger}$. 
Following condition \eqref{eq:tbar_cond_3_fin_vis} and via induction, we have $\Vv_i(x_0,z_j) \leq  \Vv_i(x_0, z_0) - j\xi_i $ for all $j \in \N_{ \geq 1}$.
Thus there exists some $j \in \N$, such that $\Vv_i(z_j) \leq  \Vv_i^* - j\xi_i < \Vv_i^{\dagger} $ which is a contradiction.
\qed
\end{proof}

Observe that one may consider the set $\Xx_{VF}$ to not be partitioned, in which case one recovers the standard conditions of well-foundedness as in \cite[Definition 3.2]{murali2024closure}.
In the following section, we define  $\text{C}^3$s to synthesize a controller for a control system $\Sys = ( \Xx, \Xx_0, U, f)$ to show a set is visited infinitely often.

\subsection{C$^3$s for Infinite Visits}
\label{subsec:ccc_buchi}
In this section, our objective is to design a controller to ensure that the system visits a set $\Xx_{INF} \subseteq \Xx$ infinitely often via $\text{C}^3$s.
We discuss a few approaches to do so, and motivate the use of each successive approach by discussing how they tackle issues with respect to the other.
Overall, the three approaches we consider are as follows:
\begin{itemize}
    \item [1.] We first consider a standard ranking function over transition invariants (Definition \ref{def:ccc_inf_vis_rec} following \cite{wang2025verifying} .
    \item [2.] Unfortunately, a standard ranking function is not disjunctively well-founded. Thus we introduce proofs that are disjunctive (Definition \ref{def:ccc_inf_vis_sing}) but this relies on lookaheads.
    \item [3.] We introduce conditions dependent on  counters rather than lookaheads in Definition \ref{def:ccc_inf_vis_counter}, when the number of states that are not in $X_{INF}$  between successive visits to $X_{INF}$ increases.
    \item  [4.] When the number of states that are not in $X_{INF}$  between successive visits to $X_{INF}$ increases up to a threshold, and then remains bounded, one can use an approach similar to bounded model checking (Definition \ref{def:ccc_inf_vis_counter_fin_piece}).
\end{itemize}

An initial approach to design controllers is to modify the notion of closure certificates for recurrence used in \cite{wang2025verifying} by considering the control input as follows.
\begin{definition}
    \label{def:ccc_inf_vis_rec}
    Consider a control system $\Sys = (\Xx, \Xx_0, U, f)$ and a set of states $\Xx_{INF}$ that must be visited infinitely often.
    Then a bounded function $\Tt: \Xx \times \Xx \to \R$, is a control closure certificate for recurrence, if there exists $\xi \in \R_{ > 0}$ such that for all $x \in \Xx$ there exists $u \in U$ such that:
    \begin{align}
        &\big( \Tt(x, f(x,u)) \geq 0 \big), \label{eq:tbar_cond_1_inf_vis_recur}
    \end{align}
    and for all $x, y \in \Xx$, and for all $u \in U$, we have:
    \begin{align}
        & \big(  \Tt(x, f(x,u) \big) \geq 0 \big) \implies  \Big( \big( \Tt(f(x,u), y) \geq 0 \big) \implies \big( \Tt(x, y) \geq 0 \big) \Big), \label{eq:tbar_cond_2_inf_vis_recur}
    \end{align}
    and for all $x_0 \in \Xx_0$, $y \in \Xx \setminus \Xx_{INF} $, and $u \in U$ we have:
    \begin{align}
            &\big( \Tt(x_0, y) \geq 0 \big) {\wedge} \big( \Tt(y,f(y,u)) \geq 0 \big) \big) {\implies} \big( \Tt(x_0, f(y,u) ) \geq \Tt(x_0,y) + \xi \big).  \label{eq:tbar_cond_3_inf_vis_recur} 
    \end{align}
\end{definition}
\begin{theorem}
\label{thm:ccc_inf_vis_rec}
 Consider a control system $\Sys = (\Xx, \Xx_0, U, f)$ and a set $\Xx_{INF}$ that must be visited infinitely often.
The existence of a function $\Tt$ as in Definition \ref{def:ccc_inf_vis_rec} guarantees that there exists a controller $\kappa$ to ensure that the set $\Xx_{INF}$ is visited infinitely often.
\end{theorem}
We omit the proof of Theorem \ref{thm:ccc_inf_vis_rec} as it follows in a similar fashion as the proof of \cite[Theorem 2]{wang2025verifying}.
While $\text{C}^3$s as in Definition \ref{def:ccc_inf_vis_rec} provide an effective automated approach to synthesize a controller that ensures a set is visited infinitely often, they unfortunately face two drawbacks.
First, there are some systems for which one is unable to satisfy conditions \eqref{eq:tbar_cond_1_inf_vis_recur}-\eqref{eq:tbar_cond_3_inf_vis_recur} even though the set $\Xx_{INF}$ is visited infinitely often.
Second, the above does not rely on transition invariants to provide well-founded arguments and thus
one may provide a similar argument with other invariants such as barrier certificates as illustrated in \cite{chatterjee2024sound}.
The key benefit of this approach then relies on the expressivity of the transition invariant, compared to the state invariant.
To illustrate the first drawback we consider the following Lemma.
\begin{lemma}
    \label{lem:recur_not_complete}
    There exists a dynamical system $\Sys_{dis} = (\Xx, \Xx_0, f)$, that visits a region $\Xx_{INF} \subseteq \Xx$ infinitely often, however one is unable to find a control closure certificate for recurrence as in Definition \ref{def:ccc_inf_vis_rec}. 
\end{lemma}
The proof for Lemma \ref{lem:recur_not_complete} is found in Appendix \ref{ap:lem_nc} but this relies on systems whose state sequence is similar to the one described in Figure \ref{fig:inf_vis_chain}.
The key issue is that we cannot assume the function $\Tt$ to be bounded both from above and below.
To provide a (relatively)-complete proof rule, one can only assume the function to be bounded in one direction as in \cite{chatterjee2024sound}.
An easy remedy to the above problem is to change condition \eqref{eq:tbar_cond_3_inf_vis_recur} by considering a function $\Vv$ to denote a ranking function and ensure it is only bounded from below. 
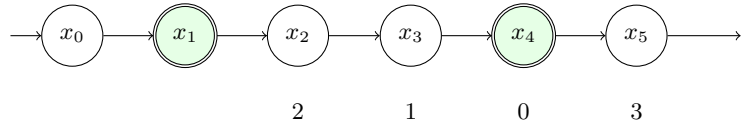
\begin{figure}
    \centering
    \begin{tikzpicture}[node distance =1cm]
    \node[initial, state, draw, initial text =] (0) at (0,0) {$x_0$};
    \node[,state, fill=green!10!white, accepting] (1) at (1.5,0) {$x_1$};
    \node[,state] (2) at (3,0) {$x_2$};
        \node[] (5) at (3,-1) {$2$};
    \node[state] (3) at (4.5,0) {$x_3$};
        \node[] (6) at (4.5,-1) {$1$};
        \node[state, fill=green!10!white, accepting] (4) at (6,0) {$x_4$};
            \node[] (7) at (6,-1) {$0$};
            \node[state,] (8) at (7.5,0) {$x_5$};
                    \node[] (10) at (7.5,-1) {$3$};
        \node[] (9) at (9,0) {};
    \path[->]
    (0) edge node[above]{} (1)
    (1) edge[] node[above]{} (2)
    (2) edge node[]{} (3)
   (3) edge node[]{} (4) 
    (4) edge node[]{} (8)
    (8) edge node[]{} (9);
    \end{tikzpicture}
    \caption{An infinite chain that visits $\Xx_{INF}$ infinitely often. However in such a chain, the distance between successive visits to this set increases.
    We annotate the least possible value of $\Tt(x_0, x_i)$ below each state $x_i$ in this example, when $\xi = 1$. Observe that the value of $\Tt(x_0, x_i)$ needs to increase (unboundedly) as the sequence grows.
    }
    \label{fig:inf_vis_chain}
\end{figure}

One benefit of $\text{C}^3$s is that one may still find a bounded function for the above example, even when the distance between successive visits to the set $\Xx_{INF}$ increases.
To do so, we need to modify the conditions in Definition \ref{def:ccc_inf_vis_rec} and rely on the following insight.
The key issue with the above example is that the value of $\xi$ needs to be fixed.
While we cannot arbitrarily change $\xi$ at every step, we note that one can select the value $\xi$ as a parameter for each subsequence between states in $\Xx_{INF}$.
Allowing this value to be a function of the possible state in $\Xx_{INF}$ that may be visited in the future ensures that one may still find a bounded function to act as a proof of well-foundedness.
We rely on this to consider a notion of C$^3$s which allows for the value of $\xi$ to be dependent on a state as defined below.

\begin{definition}
     \label{def:ccc_inf_vis_sing}
     Consider a control system $\Sys = (\Xx, \Xx_0, U, f)$ and a set of states $\Xx_{INF}$ that must be visited infinitely often.
    Then, function $\Tt: \Xx \times \Xx \to \R$, and bounded (from below) function $\Vv: \Xx \times \Xx \times \Xx \to \R_{\geq 0}$ are a control closure certificate, if, for all $x \in \Xx$ there exists $u \in U$ such that:
     \begin{align}
        &\big( \Tt(x, f(x,u)) \geq 0 \big), \label{eq:tbar_cond_1_inf_vis_sing}
     \end{align}
     and for all $x, y \in \Xx$, for all $u \in U$, we have:
    \begin{align}
         & \big(  \Tt(x, f(x,u) \big) \geq 0 \big) \implies \Big( \big( \Tt(f(x,u), y) \geq 0 \big) \implies \big( \Tt(x, y) \geq 0 \big) \Big), \label{eq:tbar_cond_2_inf_vis_sing}
     \end{align}
     and for all $x_0 \in \Xx_0$, and all $y \in \Xx \setminus \Xx_{INF}$, there exists $w \in \Xx_{INF}$ and $\xi \in \R_{ > 0}$, such that for all $z,z' \in \Xx \setminus \Xx_{INF}$ we have:
     \begin{align}
             &\big( \Tt(x_0, y) \geq 0 \big) {\wedge} \big( \Tt(y,z) \geq 0 \big)   {\wedge} \big( \Tt(z, z') \geq 0 \big) {\implies}  \nonumber \\ &  \qquad \qquad \big( \Vv(y,z',w) \leq \Vv(y,z,w) {-} \xi \big).   \label{eq:tbar_cond_3_inf_vis_sing}
    \end{align}
 \end{definition}

Note, however, that such a strategy introduces an additional 
quantifier alternation in condition \eqref{eq:tbar_cond_3_inf_vis_sing} between $y$ and $w$, and $w$ and $z$,  respectively and that the positive real $\xi$ depends on the universally quantified value $y$.
Intuitively, one seeks to select the value of $\xi$ such that state $w$ that is reachable from both states $z$, and $z'$ with a ranking function argument that decreases with $\xi$.
Second, we need to consider every state $z$, and $z'$ even if there might be accepting states that are visited in between.
Returning to the example from Figure \ref{fig:inf_vis_chain}, one needs to provide a ranking function argument between states $x_5$ and $x_3$, even though there exists an accepting state between them.
Typically, existing approaches such as in \cite{abate2024stochastic,chatterjee2020polynomial} avoid this issue by allowing for the ranking function to increase over the accepting states.
The issue one faces is that in examples such as in Figure \ref{fig:inf_vis_chain} where one jumps to, and out of, the set $\Xx_{INF}$ without significant dwell time, one require values of function $\Vv$ to change significantly.
Such an increase condition is not immediately obvious for closure certificates without considering one step successors as in Definition \ref{def:ccc_inf_vis_rec}.
Due to these challenges, we consider another set of conditions for control closure certificates.
Given a state sequence of a system that visits $\Xx_{INF}$ infinitely often, consider a partitioning of the state sequence into subsequences in between successive visits.
For example, consider the infinite chain in Figure \ref{fig:inf_vis_chain}, and consider the subsequences $\langle x_0, x_1\rangle$, $\langle x_2, x_3, x_4 \rangle$, and so on.
Observe that if one were able to find ranking functions for each of these subsequences, then one is able to provide a proof that the system visits the set $\Xx_{INF}$ infinitely often.
Formally, consider the control system defined as $\hat{\Sys} = \{\hat{\Xx}, \hat{\Xx}_0, U, \hat{f} \}$, where $\hat{\Xx} = \Xx \times \R$ denotes the set of states,  $\hat{\Xx}_0 = \Xx_0 \times \{ 0\} $ denotes the set of initial states, and the transition function $\hat{f}((x,j),u) = \{ (f(x,u),k)  \}$, where $k = j+1$ if $x \in \Xx_{INF}$ and $k = j$ otherwise \footnote{This construction is similar to the counter construction to degeneralize a generalized B\"uchi automaton \cite{thomas2002automata}, except here we consider the counter to be unbounded}.
An illustration of this construction on the infinite sequence in Figure \ref{fig:inf_vis_chain} can be found in Appendix \ref{ap:inf_vis_counter}.
Intuitively, this modified system consists of the states of original system appended with a counter.
The value of this counter increments every time one visits the set $\Xx_{INF}$.
Then one may find different ranking function $\Vv_j$ that show states in $\Xx \setminus \Xx_{INF}$ are visited only finitely often for every counter value $j$.
In the next definition, we consider how one may find certificates over this modified system to ensure that the set $\Xx_{INF}$ is visited infinitely often.

\begin{definition}
    \label{def:ccc_inf_vis_counter}
    Consider a control system $\Sys = (\Xx, \Xx_0, U, f)$ and a set of states $\Xx_{INF}$ that must be visited infinitely often.
    Then, function $\Tt: (\Xx \times \R) \times (\Xx \times \R) \to \R$, and a function $\Vv: \Xx \times \Xx \times \R \to \R_{\geq 0}$ that is bounded from below in $\Xx \times \Xx$ for every $\R$, are a control closure certificate, if for all $x \in \Xx$, there exists $u \in U$ such that for all $j \in \R$:
    \begin{align}
        &\big( \Tt( (x,j), (f(x,u), k) ) \geq 0 \big), \label{eq:tbar_cond_1_inf_vis_counter}
    \end{align}
    where $k = j+1$ if $x \in \Xx_{INF}$, and $k = j$ otherwise.
    And for all $x, y \in \Xx$, $j \in \R$, $\ell \geq (j+1)$, and for all $u \in U$, we have:
    \begin{align}
        & \big(  \Tt( (x,j), (f(x,u) ,k) \big) \geq 0 \big) \implies \nonumber \\
        & \Big( \big( \Tt( (f(x,u), k), (y, \ell)) \geq 0 \big) \implies \big( \Tt((x,j), (y, \ell)) \geq 0 \big) \Big), \label{eq:tbar_cond_2_inf_vis_counter}
    \end{align}
    where $k = j+1$ if $x \in \Xx_{INF}$, and $k = j$ otherwise.
    And for all $j \in \R$, there exists $\xi_j \in \R_{ > 0}$ such that for all  $x_0 \in \Xx_0$, and for all $z,z' \in \Xx \setminus \Xx_{INF}$, we have:
    \begin{align}
            &\big( \Tt((x_0,0),(z,j)) \geq 0 \big) \wedge \big( \Tt((z,j),(z',j)) \geq 0 \big) \implies  \nonumber \\
         &\big( \Vv(x_0,z', j) \leq \Vv(x_0,z, j) - \xi_j \big).   \label{eq:tbar_cond_3_inf_vis_counter} 
    \end{align}
\end{definition}
\begin{theorem}
    \label{thm:ccc_inf}
    Consider a control system $\Sys = (\Xx, \Xx_0, U, f)$ and a set $\Xx_{INF}$ that must be visited infinitely often.
    The existence of functions $\Tt$ and $\Vv$ as in Definition \ref{def:ccc_inf_vis_counter} guarantees that there exists a controller $\kappa$ to ensure that $\Xx_{INF}$ is visited  infinitely often.
\end{theorem}
The proof of Theorem \ref{thm:ccc_inf} can be found in Appendix \ref{ap:Thm_cc_inf_proof}.
Observe that while the above definition still has a quantifier alternation, the existential quantifier is no longer over the states of the system or the set of control inputs, but rather over the selection of the real value $\xi$ for each choice of counter value $j$.
Intuitively, the goal of the function $\Vv$ is to prove that only finitely many states are visited with a counter value $j$ for any real value $j \in \R$.
Observe that each of these proofs for the counter value $j$ corresponds to a proof where the goal is to show a set of states is visited only finitely often.
Thus, one may adapt the approach considered in Definition \ref{def:ccc_fin_vis} to provide disjunctive well-founded proofs for each counter value.
Here, we may replace condition \eqref{eq:tbar_cond_3_inf_vis_counter} by combining it with condition \eqref{eq:tbar_cond_3_fin_vis}.
We describe these conditions in Appendix \ref{ap:ccc_disjunct_inf}. 
To avoid considering all real values for the counter $j$, one may adopt a strategy similar to bounded model checking where we first fix a bound (say $j_{\max}$) on the value $j$.
Then we consider piecewise functions for conditions \eqref{eq:tbar_cond_1_inf_vis_counter} to \eqref{eq:tbar_cond_3_inf_vis_counter} for all $0 \leq j \leq j_{\max}$.
The benefit of this approach lies in the fact that the control input $u$ can now depend on the value of the counter $j$.
If we changed the alternation in condition \eqref{eq:tbar_cond_1_inf_vis_counter}, then we would need an infinite memory policy as the value $u$ could depend on the value of the unbounded counter $j$.
In particular, if we are able to show that we have that $\Tt((x_0, 0), (z, j_{\max} +1 )) \geq 0 \implies \Tt((x_0, 0), (z, j_{\max} ) ) \geq 0$  for any $x_0 \in \Xx_0$ and $z \in \Xx$, we note that one may use the same function $\Vv_{j_{\max}}$ and the same control inputs $u$ as that of counter $j_{\max}$.
In such a case, the set of states $z$ that satisfy the above conditions represent an overapproximation of the states that can be reached infinitely often.
In particular, this provides guarantees for systems which see both accepting and non-accepting states infinitely often. 
We described the conditions for such a certificate below.

\begin{definition}
\label{def:ccc_inf_vis_counter_fin_piece}
    Consider a control system $\Sys = (\Xx, \Xx_0, U, f)$, a set of states $\Xx_{INF}$ that must be visited infinitely often and $j_{\max} \in \N$ denote a bound on the counter value.
    Then, functions $\Tt_{j,\ell}: \Xx \times \Xx \to \R$, and functions $\Vv_j: \Xx \times \Xx \to \R_{\geq 0}$ for all $0 \leq j \leq \ell \leq j_{\max}$, that are bounded from below in
    $\Xx \times \Xx$, constitute a control closure certificate, if for all $0 \leq j \leq j_{\max}$, and for all $x \in \Xx$, there exists $u \in U$ such that:
    \begin{align}
        &\big( \Tt_{j, k}( x, f(x,u) ) \geq 0 \big), \label{eq:tbar_cond_1_inf_vis_piece}
    \end{align}
    where $k = j+1$ if $x \in \Xx_{INF}$, and $k = j$ otherwise. 
    And for all $0 \leq j < \ell \leq  (j_{\max}+1)$, and for all $x, y \in \Xx$, and for all $u \in U$, we have:
    \begin{align}
        & \big(  \Tt_{j,k}( x, f(x,u)) \big) \geq 0 \big) {\implies}  \Big( \big( \Tt_{k, \ell}( f(x,u), y) \geq 0 \big) {\implies} \big( \Tt_{j, \ell}(x, y) \geq 0 \big) \Big), \label{eq:tbar_cond_2_inf_vis_piece}
    \end{align}
    where $k = (j+1)$ if $x \in \Xx_{INF}$, and $k = j$ otherwise.
    And for all $x_0 \in \Xx_0$, and $z, z' \in \Xx$ we have:
        \begin{align}
        & \big(\Tt_{0,(j_{\max} + 1)} (x_0, z)  ) \geq 0 \big) {\implies} \big( \Tt_{0, j_{\max}}(x_0, z ) \geq 0 \big), \text{ and } \label{eq:tbar_implication_contain_1} 
        \end{align}
    And for all $0 \leq j \leq j_{\max}$, there exists $\xi_j \in \R_{ > 0}$ such that for all  $x_0 \in \Xx_0$, for all $z,z' \in \Xx \setminus \Xx_{INF}$, we have:
    \begin{align}
            &\big( \Tt_{0, j}(x_0, z) \geq 0 \big) \wedge \big( \Tt_{j,j}(z,z') \geq 0 \big) \implies  \big( \Vv_j(x_0,z') \leq \Vv_j(x_0,z) - \xi_j \big).   \label{eq:tbar_cond_3_inf_vis_piece} 
    \end{align}
\end{definition}
\begin{theorem}
    \label{thm:ccc_inf_piece}
    Consider a control system $\Sys = (\Xx, \Xx_0, U, f)$, a set $\Xx_{INF}$ that must be visited infinitely often and $j_{\max} \in \N$ denote a bound on the counter value.
    The existence of functions $\Tt_{j, \ell}$ and $\Vv_j$ for all $0 \leq i, j \leq (j_{\max} +1 )$ as in Definition \ref{def:ccc_inf_vis_counter_fin_piece} guarantees that there exists a controller $\kappa$ to ensure that $\Xx_{INF}$ is visited  infinitely often.
\end{theorem}
\begin{proof}
Consider the finite memory control strategy $\kappa: \Xx \times \{ 0, \ldots, j_{\max} \}$ such that we use controller $\kappa(x_i, j) $ for any $j \leq j_{\max}$, if $j$ states in the set $\Xx_{INF}$ are visited, and $\kappa(x_i, j_{max})$ otherwise.
Let $\kappa(x_i, j) = \{ u \mid \Tt_{j, k} (x_i, f(x_i, u)) \geq 0 \}$, where $k = j+1$ if $x_i \in \Xx_{VF}$ and $k = j$ otherwise, when $j \leq j_{\max}$.
Assume that the system visits the set $\Xx_{INF}$ only finitely often under this control strategy.
That is, let the corresponding sequence be $ \bm{x} = \langle x_0, x_1, \ldots  \ldots  \rangle$.
Let this correspond to the sequence  $\bm{\hat{x}} = \langle (x_0, \ell_0), (x_1, \ell_1), \ldots  \ldots  \rangle$ in the system $\hat{\Sys}$,
where $x_{i+1} = \kappa(x_i, \ell_i ) $, and $\ell_{i+1} = \ell_i + 1$ if $x_i \in \Xx_{INF}$ and $\ell_{i+1} = \ell_i$ otherwise.
Observe that $\ell_0 = 0$.
Thus there exists some $k \in \N$ such that for all $i \geq k$, we have $\ell_i = \ell_k$. 
Following the results of Theorems \ref{thm:ccc_fin_vis} and \ref{thm:ccc_inf}, we cannot have $k \leq j_{\max}$.
Let $k \geq (j_{\max}+1)$, and consider state $x_\upsilon$
such that $\ell_\upsilon = j_{\max}$, and $\ell_{\upsilon+1} = (j_{\max} + 1)$.
Following condition \eqref{eq:tbar_cond_1_inf_vis_piece}, we must have $\Tt_{\ell_{(\upsilon - 1)}, \ell_{\upsilon}  } (x_{\upsilon-1}, x_\upsilon) \geq 0 $ and  $\Tt_{ \ell_{\upsilon},\ell_{ (j_{\max} + 1)}   } (x_{\upsilon}, x_{\upsilon+1}) \geq 0 $.
Thus via condition \eqref{eq:tbar_cond_2_inf_vis_piece}, we must have $\Tt_{\ell_{(\upsilon - 1)}, (j_{\max} + 1) } (x_{\upsilon-1}, x_{\upsilon+1} ) \geq 0 $.
Inducting on conditions \eqref{eq:tbar_cond_1_inf_vis_piece} and \eqref{eq:tbar_cond_2_inf_vis_piece}, we get $\Tt_{0, (j_{\max} + 1) } (x_{0}, x_{\upsilon+1} ) \geq 0$ and so following condition \eqref{eq:tbar_implication_contain_1}, we must have $\Tt_{0, j_{\max}} (x_0, x_{\upsilon + 1}) \geq 0$.
Observe that one may select a control input $u$ as in condition \eqref{eq:tbar_cond_1_inf_vis_piece} to ensure $\Tt_{j_{\max},j_{\max+1} } (x_{\upsilon +1 }, x_{\upsilon + 2})$, and thus we have an inductive argument that for any $r \in \N_{\geq 1}$, we have   $\Tt_{0, j_{\max}} (x_0, x_{\upsilon + r}) \geq 0$.
Thus, we observe that the antecedent of condition \eqref{eq:tbar_cond_3_inf_vis_piece} holds and so the ranking function must decrease.
In a similar manner as earlier proofs we can conclude that this creates a contradiction.
\end{proof}

Now, we describe how one may use the above conditions to design controllers against $\omega$-regular specifications.
\subsection{C$^3$s for $\omega$-regular Specifications}
\label{subsec:ccc_parity}
To synthesize controllers against $\omega$-regular specifications via C$^3$s, let the DPA $\Aa= (\Sigma, Q, q_0, \delta, {Acc})$ denote the desired specification and the set $ \{1, \ldots ,c \}$ denote the set of priorities or colors, \textit{i.e.}, $Acc: Q \to \{1, \ldots c \}$.
Then the system $\Sys$ under labeling map $\Ll$ and controller $\kappa$ satisfies the $\omega$-regular specification if $\Tr \subseteq  L(\Aa)$, \textit{i.e.}, every trace of the system under the labeling map and controller $\kappa$ is accepted by $\Aa$.
To synthesize controller $\kappa$, we  first construct the product $\Sys \otimes \Aa = (\Xx', \Xx_0', U, f')$ of the system $\Sys = (\Xx, \Xx_0, U, f)$ with the DPA $\Aa$ representing the  the specification, where 
$\Xx' = \Xx \times Q$ indicates the state set, and
$\Xx'_0 = \Xx_0' \times q_0$ are the initial set of states.
We define the state transition relation $f': \Xx \times Q \times U \to \Xx \times Q$ as :
 \[ f'((x,q_i), u) = \big\{ (f(x,u), q_j) \mid q_j \in \delta(q_i, \Ll(x)) \big\}. \]

We note that ensuring the specification is satisfied is equivalent to showing that the minimum priority seen infinitely often is even.
That is, we need to show that for every sequence $\bm{s}' = \langle x'_0, x'_1 , \ldots \rangle$, under controller $\kappa$, if $x'_j = (x_j, q_i)$ where $Acc(q_i) $ is odd, either: 
\begin{itemize}
    \item [1.] There exists some $\ell_1 \in \N$ such that for all $\ell_2 \geq \ell_1$, we have $\pi_2(x'_{\ell_2}) \neq q_i$, or
    \item [2.] There exists some $q_j$, where $Acc(q_j)$ is even, $Acc(q_j) < Acc(q_i)$, and for all $\ell_1 \in \N$, there exists $\ell_2 \geq \ell_1$, such that we have $\pi_2(x'_{\ell_2}) = q_J$.
\end{itemize}
These correspond to showing that either we see automata states with priority $Acc(q_j)$ only finitely often or we see some automaton state $q_i$ which as an even priority that is less than $Acc(q_j)$ infinitely often.
We observe that the first two conditions for building invariants are the same for all the conditions.
To do so, let us assume sets $\Xx_{VF}$ and $\Xx_{INF}$ are sets that must be visited finitely often and infinitely often.
Then one may search for certificates to prove both of the above as follows:

\begin{definition}
    \label{def:ccc_fin_inf}
    Consider a control system $\Sys = (\Xx, \Xx_0, U, f)$, a set of states $\Xx_{VF} = \underset{1 \leq r \leq p}{\bigcup} \Xx_{VF_r}$ that must be visited finitely often, and a set of states $\Xx_{INF}$ that must be visited infinitely often.
    Then, function $\Tt_{i,j}: \Xx \times \Xx \to \R$, functions $\Zz_r: \Xx \times \Xx \to \R_{\geq 0}$ that are bounded from below, and functions 
    $\Vv_i: \Xx \times \Xx \to \R_{\geq 0}$ that are bounded from below and defined for all $1 \leq r \leq p$ , and all $0 \leq i, j \leq j_{\max}$
    are a control closure certificate, if for all $x \in \Xx$, and all $0 \leq j \leq j_{\max}$, there exists $u \in U$ such that:
    \begin{align}
        &\big( \Tt_{j, k} (x, f(x,u) ) \geq 0 \big), \label{eq:tbar_cond_1_inf_fin}
    \end{align}
    where $k = j+1$ if $x \in \Xx_{INF}$, and $k = j$ otherwise.
    And for all $x, y \in \Xx$, and all $0 \leq j < \ell \leq (j+1)$, and for all  $u \in U$, we have:
    \begin{align}
        & \big(  \Tt_{j,k}( x, f(x,u) ) \big) \geq 0 \big) {\implies} \Big( \big( \Tt_{k, \ell}( (f(x,u), y) \geq 0 \big) {\implies} \big( \Tt_{j, \ell}(x, y) \geq 0 \big) \Big), \label{eq:tbar_cond_2_inf_fin}
    \end{align}
    where $k = j+1$ if $x \in \Xx_{INF}$, and $k = j$ otherwise.
    And for all $x_0 \in \Xx_0$, and $z, z' \in \Xx$ we have:
        \begin{align}
        & \big(\Tt_{0,(j_{\max} + 1)} (x_0, z)  ) \geq 0 \big) {\implies} \big( \Tt_{0, j_{\max}}(x_0, z ) \geq 0 \big). \label{eq:tbar_inv_omega} 
\end{align}
    And for all $x_0 \in \Xx_0$, and all $1 \leq r \leq p$, and all $0 \leq \ell_1 \leq \ell_2 \leq (j_{\max}+1)$, there exists $\xi_r \in \R_{ > 0}$ such that for all $z,z' \in \Xx_{VF_{r}}$,:
    \begin{align}
            &\big( \Tt_{0, \ell_1} (x_0, z) \geq 0 \big) \wedge \big( \Tt_{\ell_1, \ell_2} (z,z') \geq 0 \big) \implies \big( \Zz_r(x_0, z') \leq \Zz_r(x_0,z) - \xi_r \big).   \label{eq:tbar_cond_3_fin_visit_omega}
    \end{align}
    And for all $0 \leq j \leq j_{\max}$, there exists $\xi_j \in \R_{ > 0}$ such that for all  $x_0 \in \Xx_0$, for all $z,z' \in \Xx \setminus \Xx_{INF}$, we have:
    \begin{align}
            &\big( \Tt_{0,j} ( x_0, z) \geq 0 \big) \wedge \big( \Tt_{j,j} (z,z') ) \geq 0 \big) \implies  \big( \Vv_j(x_0,z') \leq \Vv_j(x_0,z) - \xi_j \big).   \label{eq:tbar_cond_3_inf_vis_omega} 
    \end{align}
\end{definition}
\begin{theorem}
    \label{thm:ccc_inf_fin}
    Consider a control system $\Sys = (\Xx, \Xx_0, U, f)$ and a set $\Xx_{INF}$ that must be visited infinitely often and a set $\Xx_{VF} = \underset{1 \leq i \leq r}{\bigcup} \Xx_{VF_r}$ that must be visited only finitely often.
    Let $j_{\max} \in \N$ denote a threshold for the counter value.
    Then, the existence of functions $\Tt_{j, \ell}$, $\Vv_j$, and $\Zz_r$ for all $1 \leq i \leq r$, and $0 \leq j, \ell \leq (j_{\max}+1)$ as in Definition \ref{def:ccc_fin_inf} guarantees that there exists a controller $\kappa$ to ensure that $\Xx_{INF}$ is visited  infinitely often and $\Xx_{VF}$ only finitely often.
\end{theorem}

The proof of Theorem \ref{thm:ccc_inf_fin} follows from the proof of Theorem \ref{thm:ccc_fin_vis} and Theorem \ref{thm:ccc_inf_piece}.
Note that one may equivalently consider other conditions for finite and infinite visits as discussed earlier.
We now show how one may use certificates as in Definition \ref{def:ccc_fin_inf} to design controllers for $\omega$-regular specifications.
First, we search for a C$^3$ over $\Sys \otimes \Aa$ to ensure that the set $\Xx'_{VF} = \Xx \times \{ q_i \mid Acc(q_i) = 1 \}$ is visited only finitely often.
If we fail to do so, then we do not have an even priority that is smaller than $1$ as the set of colors we consider are $\{1, \ldots c \}$ and thus we need to change the template of certificates.
If we are successful, we then try to find a certificate to ensure that the set $\Xx'_{VF} = \Xx \times \{ q_i \mid Acc(q_i) = 1 \} \cup  \Xx \times \{ q_i \mid Acc(q_i) = 3 \}$ is visited only finitely often.
If we fail to do so we instead seek to find two certificates to show that $\Sys \otimes \Aa$ visits  $\Xx'_{VF} = \Xx \times \{ q_i \mid Acc(q_i) = 1 \}$ only finitely often and visits the set $\Xx'_{INF} = \Xx \times \{ q_i \mid Acc(q_i) = 2 \}$ 
infinitely often.
If we succeed, we continue till we fail for some state with odd parity $q_j$, we then set $\Xx_{INF} = \Xx_{INF} \cup \Xx \times \{q_\ell \}$ where $Acc(q_\ell) < Acc(q_j)$ and  $\ell$ is even.
We continue this process until we either find a certificate satisfying the above conditions for all odd automaton states, or we conclude that we do not have a proof with the desired template.
We should note that we need to search independently over these different combinations of priorities as we consider different sets for $\Xx_{INF}$ and $\Xx_{VF}$ respectively.

\section{Computation of C$^3$s}
\label{sec:comp}
To find control closure certificates, we make use of a semidefinite programming approach~\cite{Parrilo_2003} via sum-of-squares (SOS) similar to the automated search for standard barrier certificates~\cite{prajna_2004_safety}.
A  set $A \subseteq \R^n$ is semi-algebraic if it  can be defined with the help of a vector of polynomial inequalities $h(x)$ as $A = \{ x \mid h(x) \geq 0 \}$, where the inequalities is interpreted component-wise.

To adopt an SOS approach to find C$^3$s, we consider the sets $\Xx$, $\Xx_0$, $\Xx_{VF}$, $\Xx_{INF}$, $\Xx \setminus \Xx_{INF}$ and $U$ to be semi-algebraic sets defined with the help of vectors of polynomial inequalities $g$, $g_0$, $g_{VF}$, $g_{INF}$, $g_{CINF}$, and $g_{u}$, respectively.
When dealing with finite partitions such as in condition \eqref{eq:tbar_cond_1_fin_vis}, we also assume that each partition $\Xx_{VF_i}$ is represented by polynomial inequalities $g_{VF_i}$ respectively.
As these sets are semi-algebraic, we know that sets $\Xx \times \Xx \times U$, $\Xx_0 \times \Xx$, $\Xx_0 \times \Xx_{VF_i} \times \Xx_{VF_i}$ and $\Xx_0 \times (\Xx \setminus \Xx_{INF}) \times (\Xx \setminus \Xx_{INF})$ are semi-algebraic as well, with their corresponding vectors denoted by $g_{A}$, $g_{B}$, $g_{C,i}$ and $g_{D}$ respectively.
Furthermore, we assume that the function $f$ is polynomial.
To deal with constraints with implications, we rewrite them in the form of sufficient conditions.
We strengthen implication-based conditions into ones that are compatible with SOS optimization via S-procedure \cite{yakubovich1971s}.
For example, condition \eqref{eq:tbar_cond_2_fin_vis_alt} can be rewritten as $\Tt(x, y) - \tau \Tt(f(x,u), y) \geq 0$, where $\tau \in \mathbb{R}_{>0}$. Observe that if the above holds, then so does condition \eqref{eq:tbar_cond_2_fin_vis_alt}.
Observe that if one satisfies this inequality, then condition \eqref{eq:tbar_cond_2_fin_vis_alt} holds.
To find C$^3$s, we fix the template to be a linear combination of user-defined basis functions of the form: $\Tt(x,y) = \mathbf{c}_{\Tt}^T\mathbf{b}(x,y) = \sum_{i = 1}^{n} c_i b_i(x,y)$, 
   $\Vv(x) = \mathbf{c}_{\Vv}^T\mathbf{b}(x) = \sum_{m = 1}^{n} \bar{c}_m b_m(x)$, and 
    $\Zz_i(x) = \mathbf{c}_{\Zz, i}^T\mathbf{b}(x) = \sum_{m = 1}^{n} \underline{c}_{m,i} b_m(x)$,
where functions $b_i(x,y)$ are monomials over the state variables $x$ and $y$, and $b_m(x)$ are monomials over the state variable $x$, and $c_1, \ldots, c_n$, $\bar{c}_1, \ldots, \bar{c}_n$, and $\underline{c}_{1,i}, \ldots, \underline{c}_{n,i}$ are real coefficients.
We then restrict the input set $U_d = \{u_1, \ldots, u_m \}$ to consist of a set of finite inputs.
One can then rewrite condition~\eqref{eq:tbar_cond_1_fin_vis} as $\sum_{i = 1}^{m}\mu_i \Tt(x, f(x,u_i)) \geq 0$ where $\mu_i \in \R_{>0}$.
Observe that satisfying this implies condition~\eqref{eq:tbar_cond_1_fin_vis}.
Now we illustrate how one can encode conditions~\eqref{eq:tbar_cond_1_inf_fin}-\eqref{eq:tbar_cond_3_inf_vis_omega} of Definition \ref{def:ccc_fin_inf} into SOS program.
To ensure functions $\Vv_j$ and $\Zz_r$ are bounded from below, we assume them to be SOS over the relevant sets.
One can reduce the search for a C$^3$ to showing that the following polynomials are SOS for all states $x,y \in \Xx$, $x_0 \in \Xx_0$, $w_r,w_r' \in \Xx_{VF_r}$ for any $0 \leq r \leq p$, and all $z,z' \in \Xx\setminus \Xx_{INF}$, and all inputs $u \in U$, and finite inputs $u_t \in U_d$, and constants $j,\ell_1 \in \set{1,\cdots,j_{max}}$, $\ell \geq (j+1)$, $\ell_2 \geq (\ell_1+1)$, where $k = (j+1)$ if $x \in \Xx_{INF}$, and $k = j$ otherwise:
\begin{align}
    &\Tt_{j,k}(x,f(x,u_m))+\sum_{t=1}^{m-1}\big(\mu_t\Tt_{j,k}(x,f(x,u_t))\big)-\lambda_0^T(x)g(x),  
    \label{eq:sos_tbar_1} \\
    &\Tt_{j,\ell}(x,y) - \tau_1\Tt_{k,\ell}(f(x,u),y)- \tau_2\Tt_{j,k}(x,f(x,u)) -\lambda_{A}^T(x,y,u)g_{A}(x,y,u), 
    \label{eq:sos_tbar_2} \\
     &\Tt_{0, j_{\max}}(x_0, y ) - \tau_3 (\Tt_{0,(j_{\max} + 1)} (x_0, y)) - \lambda_B^T(x_0,y)g_B(x_0,y),
    \label{eq:sos_cbar_7} \\
    &\Zz_{r}(x) - \lambda_{1}^T(x)g(x),
    \label{eq:sos_cbar_3}\\
    &\Zz_{r}(w_r) {-} \Zz_{r}(w_r') {-} \tau_4\Tt_{0,\ell_1}(x_0,w_r) {-} \tau_5\Tt_{\ell_1,\ell_2}(w_r,w_r') {-} \xi_i \nonumber\\
    &\qquad {-}\lambda_{C,r}^T(x_0,w_r,w'_r)g_{C,r}(x_0,w_r,w_r'),
    \label{eq:sos_cbar_4}\\
    &\Vv_{j}(x) - \lambda_{2, j}^T(x)g(x), 
    \label{eq:sos_cbar_5}\\
    &\Vv_{j}(z) {-} \Vv_{j}(z') {-} \tau_6\Tt_{0,j}(x_0,z) - \tau_7\Tt_{j,j}(z,z') {-} \xi_D {-}\lambda_{D}^T(x_0,z,z')g_{D}(x_0,z,z'),
    \label{eq:sos_cbar_6}
\end{align}
where $\mu_i, \tau_i, \xi_i \in\R_{ > 0}$ are some positive values, and the multipliers $\lambda_{A}$, $\lambda_{B}$, $\lambda_{C,r}$, $\lambda_{D}$ are arbitrary sum-of-squares polynomials in their respective state variables over the regions $\Xx \times \Xx \times U$, $\Xx_0 \times \Xx$, $\Xx_0 \times \Xx_{VF_r} \times \Xx_{VF_r}$, and $\Xx_0 \times (\Xx \setminus \Xx_{INF}) \times (\Xx \setminus \Xx_{INF})$, respectively. Similarly, $\lambda_0$, $\lambda_{1}$, and $\lambda_{2,j}$ are arbitrary sum-of-squares polynomials over the region $\Xx$. 
For $\omega$-regular specifications, one can take the product of the system with the corresponding DPA and employ the techniques discussed in \cite[Section 4.2]{murali2024closure}.
We omit these details due to lack of space, and for simpler readability, but we use these in our case studies for the next section.

\section{Case Studies}
\label{sec:case_studies}

We experimentally demonstrate the utility of C$^3$s on a control dynamical system that describes an instance of Hopf bifurcation that under certain control inputs can either exhibit periodic orbits or die out over time \cite{herman2008second}. 
Wind-induced oscillations for an overhead line \cite{myerscough1973simple} are an example of such system where the states $x_1$ and $x_2$ represent the displacement of the suspended conductor from its equilibrium position.
The control input is used to regulate the value of these states so that they stay within a predetermined bound, which is denoted by the set of states that should be visited infinitely often. 
If the system were to start outside the bounds, a control input will be applied to push the states to be within bounds. 
The set of states outside the bounds (or its subset) is then designated as a set that should be visited only finitely often. 
The definition of of our system is given below.
Given a state $x = x(t)$, we use $x'$ to denote $x(t+1)$.
\begin{align}
    \label{eq:sys_eq1}
    x_1' = x_1 + T(ux_1 -x_2-x_1(x_1^2 + x_2^2)),\\
    \label{eq:sys_eq2}
    x_2' = x_2 + T(x_1 -ux_2-x_2(x_1^2 + x_2^2)),
\end{align}
where $T = 0.1s$ is the sampling time, $x^T = [x_1,x_2]$ is the state of the system, and $u$ is the control input.
We first start with a linear function in two variables and in a single variable as our parametric template for $\Tt(x,y)$ and $\Vv(x)$ (or $\Zz(x)$), respectively. 
We increase the degree of both template functions until we find the polynomials satisfying the C$^3$ conditions for finite visit only, infinite visit only and both finite and infinite visits.
We continue to increment the degree of the polynomial template up to $degree_{max} = 4$, above which the SOS fails to compute due to device memory constraints.
To solve these SOS constraints, we make use of JuMP \cite{Lubin2023} and TSSOS \cite{wang2021tssos} in Julia. Our code is available online.
\footnote{
[Online]. Available: \url{https://github.com/maoumer/CCC}
}

\subsection{Without DPA}
For the system given above, the state set, the set of initial states, finite visit states, infinite visit states, and the input set are given by $\Xx = [-0.75,1]\times[-0.75,0.75],\Xx_0 = [0.8,1]\times[-0.2,0.2]$, $\Xx_{VF} = [0.8,1]\times[0,0.75]$, $\Xx_{INF} = [-0.75,0.75]^2$, and $U = [-3,0.5]$, respectively. 
Using constants $\mu_i = 0.5$, $\tau_i = 1$ and $\xi_i = 0.1$, we were able to obtain cubic polynomial C$^3$s. Figure \ref{fig:trajectories} displays sample trajectories generated for the system along with the relevant state sets under consideration.
The successful search results of our SOS program are included in Appendix \ref{ap:ss:no_dpa}. 

\subsection{With DPA}
\begin{figure}[t]
\centering
\begin{minipage}[b]{0.55\textwidth}
    \epsfig{file=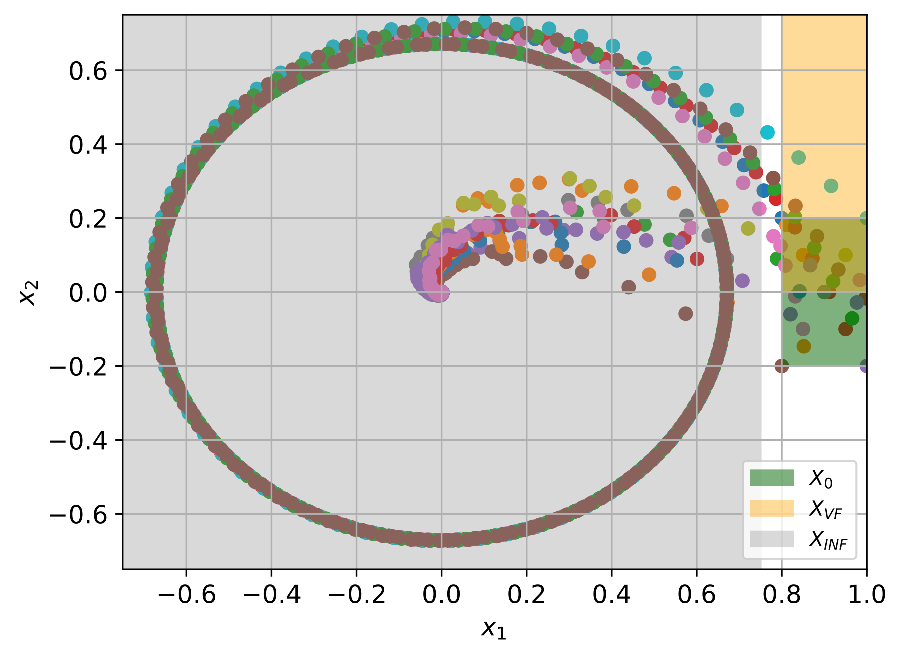, width=1.0\textwidth, keepaspectratio}
    \caption{Sample trajectories of the system along with relevant regions shaded.} 
    \label{fig:trajectories}
\end{minipage}
\hfill
\begin{minipage}[b]{0.4\textwidth}
    \begin{tikzpicture}[node distance =2cm]
    \node[initial, state, draw, initial text =,fill=green!10!white] (1) at (0,0) {$q_1$};
    \node[state, fill=blue!10!white,] (2) at (0,-3) {$q_2$};
    \node[state, fill=red!10!white,] (3) at (2.5,0) {$q_3$};
    \path[->]
    (1) edge[] node[left]{$b$} (2)
    (1) edge[] node[above]{$a$} (3)
    (2) edge[bend right] node[left]{$a$} (3)
    (2) edge[loop left] node[left]{$b$} (2)
    (3) edge[bend right] node[left]{$b$} (2)
    (3) edge[loop right] node[right]{$a$} (3);
    \end{tikzpicture}
    \caption{A DPA $\Aa$ representing the LTL formula $\mathsf{FG} a$.}
    \label{fig:aut_case_study_ltl}
\end{minipage}
\end{figure}

To demonstrate our approach for designing controllers for DPA objectives, we consider the DPA $\Aa = (\Sigma, Q ,Q_0, \delta, Acc)$ in Figure \ref{fig:aut_case_study_ltl}, where $\Sigma = \{ a,b \}$ $Q = {q_1,q_2,q_3}$, $Q_0 = {q_1}$, $Acc: Q \to \{1, \ldots, 4 \}$ such that $Acc(q_1) = 1$, $Acc(q_2) = 3$, and $Acc(q_3) = 4$.
The specification is to show that eventually, the system only witnesses states with label $a$.
This specification cannot be captured via a DBA \cite{vardi_2005_automata}.
To find a C$^3$, we consider the product of the system with the DPA. 
Here $q_1$ is the initial state, $q_2$ is the state that must be visited only finitely often, and $q_3$  only infinitely often. 
We consider a control system $\Sys = (\Xx,\Xx_0,U,f)$ where $\Xx = [-0.75,1]\times[0,0.75]$, $\Xx_0 = [0.8,1]\times[0,0.2]$, and $U = [-3,0.5]$.
We consider a labeling map $\Ll: \Xx \to \Sigma$, such that $\Ll(x) = a$ if $x \in \Xx_{INF} = [-0.75,0.75]\times[0,0.75]$, and $\Ll(x) = b$ if $x \in \Xx_{VF} = [0.75,1]\times[0,0.75]$. 
The dynamic $f$ is given in equations~\eqref{eq:sys_eq1}-\eqref{eq:sys_eq2}.
Using constants $\mu_i = 0.5$, $\tau_i = 1$ and $\xi_i = 0.01$, we were able to obtain cubic polynomial C$^3$s.
The successful search results of our SOS program are included in Appendix \ref{ap:ss:with_dpa}. 

\section{Conclusion}
\label{sec:conclusion}
We introduced a notion of control closure certificates to synthesize controllers against $\omega$-regular specifications.
We discussed different conditions to show a set is visited finitely and infinitely often and demonstrated how one may combine them for $\omega$-regular specifications and demonstrated their efficacy over relevant case studies.
As future work, we plan to investigate their use in verifying more general hyperproperties as well as possible modifications via property-directed  techniques such as IC3 \cite{bradley_2011_sat}.
 \bibliographystyle{splncs04}
\bibliography{bibliography.bib}

\appendix

\section{Proof of Lemma \ref{lem:cc_fin_vis}}
\label{ap:pr_lem_cc_fin}
We now describe the proof of Lemma \ref{lem:cc_fin_vis}.
\begin{proof}
    We prove Lemma \ref{lem:cc_fin_vis} via contradiction.
    To do so, assume that there exists an initial state $x_0 \in \Xx_0$, such that the state sequence $\bm{x} = \langle x_0, x_1, \ldots \rangle$  visits the set $\Xx_{VF}$ infinitely often.
    Following conditions \eqref{eq:tbar_cond_2_fin_vis_verif} and \eqref{eq:tbar_cond_1_fin_vis_verif} and via induction, we have $\Tt(x_0, x_i) \geq 0$ and $\Tt(x_i, x_{j}) \geq 0$ for all $i \in \N$, and all $j \geq (i+1)$.
Let $\langle y_0, y_1, \ldots \rangle$ be the subsequence that visits $\Xx_{VF}$ only finitely often.
That is the state sequence is of the form $\bm{x}_{\bm{u}} = \langle x_0, x_1, \ldots, y_0, \ldots, y_1, \ldots \rangle$.
Via Ramsey's theorem \cite{ramsey1987problem}, there exists a subsequence $\langle z_0, z_1, \ldots \rangle \in \Xx_{VF_i}$ that visits $\Xx_{VF_i}$ infinitely often for some $0 \leq i \leq p$.
From the previous results, we know that $\Tt(x_0, z_i) \geq 0$ and $\Tt(z_i, z_j) \geq 0$ for all $i \in \N$, and all $j \geq (i+1)$.
Let $ \Vv_i^* := \Vv_i(x_0,z_0) $ and as function $\Vv_i$ is bounded from below let the lower bound be $\Vv_i^{\dagger}$. 
Following condition \eqref{eq:tbar_cond_3_fin_vis_verif} and via induction, we have $\Vv_i(x_0,z_j) \leq  \Vv_i(x_0, y_0) - j\xi_i $ for all $j \in \N_{ \geq 1}$.
Thus there exists some $j \in \N$, such that $\Vv_i(z_j) \leq  \Vv_i^* - j\xi_i < \Vv_i^{\dagger} $ which is a contradiction.
\qed
\end{proof}

\section{Proof of Lemma \ref{lem:recur_not_complete}}
\label{ap:lem_nc}
We now describe the proof of Lemma \ref{lem:recur_not_complete}.
\begin{proof}
Consider a dynamical system that has a single initial state $\Xx_0 = \{ x_0 \}$ and a single state sequence $\bm{x} = \langle x_0, x_1, \ldots \rangle$ that is the infinite chain described by Figure \ref{fig:inf_vis_chain}, such that the system has only one input.
While we do not provide an explicit system in this Lemma, we describe a real-valued program that has the same trace can be found in Figure \ref{fig:prog_not_omega_reg}.
Consider a finite alphabet $\Sigma = \{a,b\}$, and a labeling map $\Ll: \Xx \to \Sigma$ that assigns a label of $a$ to every state in $\Xx \setminus \Xx_{INF}$, and $b$ otherwise.
Then the trace is of the form $ab aab aaab \ldots$.
While such a word is not $\omega$-regular, it does satisfy the desired property of ensuring that the set $\Xx_{INF}$ is visited infinitely often.
Unfortunately, one is unable to find a certificate as in conditions \eqref{eq:tbar_cond_1_inf_vis_recur}-\eqref{eq:tbar_cond_3_inf_vis_recur} to ensure this fact (even if one knows the exact reachable set).
Without loss of generality, let us assume that $\xi = 1$, and we have a function $\Tt: \Xx \times \Xx \to \R$ such that $\Tt(x,y) \geq 0$, if and only if, $y$ is reachable from $x$.
Let the maximum value attainable by $\Tt(x,y)$ be denoted as $ \Tt^*$ \textit{i.e.}, $ \Tt^* = \underset{x,y \in \Xx}{\max} \{ \Tt(x,y) \}$.
Now let us consider the the subsequence of the form $\langle x_{j}, x_{j+1}, \ldots, x_{(j+\Tt^{*})}, x_{(j+\Tt^{*}+1)} \rangle$, where $x_{(\Tt^{*}+1+1)} \in \Xx_{INF}$, and $x_i \notin \Xx_{INF}$ for any $i \in \{j, \ldots, (j+\Tt^{*}) \} $. 
Observe that such a sequence exists as the distance between subsequent visits to the set $\Xx_{INF}$ increases as the length of the sequence increases.
Following conditions \eqref{eq:tbar_cond_1_inf_vis_recur} and \eqref{eq:tbar_cond_2_inf_vis_recur} and via induction we have $\Tt(x_i, x_k) \geq 0$ for all $k \geq (i+1)$, and $i \geq j$.
Similarly, we must have $\Tt(x_0, x_i) \geq 0$ for all $i \geq j$.
Finally, following condition \eqref{eq:tbar_cond_3_inf_vis_recur}, we have that $\Tt(x_0, x_{i+1}) \leq \Tt(x_0, x_i) - 1  $ for all  $i \geq j$.
Via induction, we have that $\Tt(x_0, x_{k}) \leq \Tt(x_0, x_j) + (k-j)\epsilon  $, and thus $\Tt(x_0, x_{(j+\Tt^{*}+1)}) \leq \Tt(x_0, x_j) + (\Tt^{*}+1) $.
As function $\Tt$ is bounded, we have  $\Tt(x_0, x_{(j+\Tt^{*}+1)}) \leq 0 +  (\Tt^{*}+1) \geq \Tt^{*}$.
This is a contradiction. 
Observe that one can find a similar argument for any $\xi \in \R_{ > 0}$, as it needs to be fixed independent of the state $y \in \Xx$ in condition \eqref{eq:tbar_cond_2_inf_vis_recur}.
\qed
\end{proof}

\label{ap:system_inf}

\begin{figure}[H]
\centering
\begin{lstlisting}[language=Python, caption={A program that generates the string $abaabaaab\ldots$. Observe that this program infinitely often visits the line print(''b``).}, label={fig:prog_not_omega_reg}]
init_val = 1
while(True):
    count = init_val
    while(count >= 1):
      count = count-1
      print("a", end = ``'')
   print("b", end =``'') 
    init_val = init_val + 1
\end{lstlisting}
\end{figure}

\section{Modified counter construction of the infinite chain}
\label{ap:inf_vis_counter}
An example of our counter construction is described in Figure \ref{fig:inf_vis_chain_counter}.
\begin{figure}[h!]
    \centering
    \begin{tikzpicture}[node distance =1cm]
    \node[initial, state, draw, initial text =] (0) at (0,0) {$(x_0, 0)$};
    \node[,state, fill=green!10!white, accepting] (1) at (1.5,0) {$(x_1, 0)$};
    \node[,state] (2) at (3,-1) {$(x_2, 1)$};
    \node[state] (3) at (4.5,-1) {$(x_3, 1)$};
        \node[state, fill=green!10!white, accepting] (4) at (6,-1) {$(x_4,1)$};
            \node[state,] (8) at (7.5,-2) {$(x_5, 2 )$};
        \node[] (9) at (9,-2) {};
    \node (10) at (10.5,0) {$V_0$};
    \node (11) at (10.5,-1) {$V_1$};
    \node (12) at (10.5,-2) {$V_2$};
    \path[->]
    (0) edge node[above]{} (1)
    (1) edge[] node[above]{} (2)
    (2) edge node[]{} (3)
   (3) edge node[]{} (4) 
    (4) edge node[]{} (8)
    (8) edge node[]{} (9);
    \end{tikzpicture}
    \caption{Counter construction for the infinite chain in Figure \ref{fig:inf_vis_chain}.
    One can find a different ranking function $\Vv_i$ for every counter value $i \in \N$. 
    }
    \label{fig:inf_vis_chain_counter}
\end{figure}
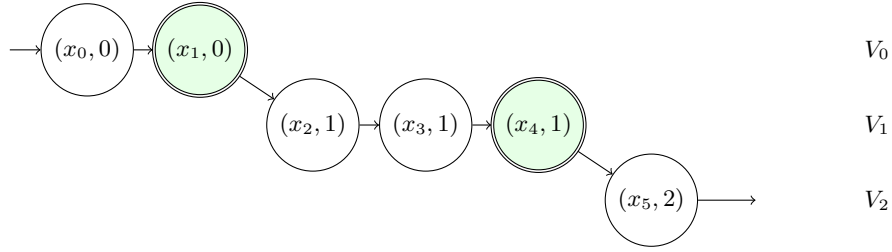

\section{Proof of Theorem \ref{thm:ccc_inf}}
\label{ap:Thm_cc_inf_proof}
\begin{proof}
    We prove Theorem \ref{thm:ccc_inf} via contradiction.
    To do so, assume that there exists an initial state $x_0 \in \Xx_0$, such that for all input sequences $\bm{u}  = \langle u_0, u_1, \ldots \rangle $, we have that the corresponding state sequence $\bm{x}_{\bm{u}} = \langle x_0, x_1, \ldots \rangle$ visits the set $\Xx_{INF}$ only finitely often.
    Consider the control input sequence selected such that $\Tt( (x_i,j), (f(x_i, u_i), k ) \geq 0$ for all $i \in \N$, where $k = j+1$ if $x_i \in \Xx_{INF}$ and $k = j$ otherwise.
    As condition \eqref{eq:tbar_cond_1_inf_vis_counter} holds, one can select a control input  in this manner for any $x_i$ in the state sequence.
    Following conditions \eqref{eq:tbar_cond_1_inf_vis_counter} and  \eqref{eq:tbar_cond_2_inf_vis_counter} and via induction, we have $
\Tt((x_0,0), (x_i, \ell_i)) \geq 0$ and $\Tt((x_i,\ell_i), (x_{j}, \ell_j) \geq 0$ for all $i \in \N$, and all $j \geq (i+1)$, and some $\ell_i \in \N$, and $\ell_j \geq \ell_i$.
Intuitively, $\ell_i$ denotes the number of times a state in $\Xx_{INF}$ is seen before state $x_i$, and similarly, $\ell_j$ denotes the number of times a state in $\Xx_{INF}$ is seen before $x_j$.
Let the state sequence stop visiting the set $\Xx_{INF}$ from index $k$, \textit{i.e.}, $\bm{x}_{\bm{u}}[k, \infty[ = \langle y_0, y_1, \ldots \rangle \in (\Xx \setminus \Xx_{INF})^{\omega}$. 
From the previous results, we know that $\Tt((x_0,0) , (y_i, k )) \geq 0$ and $\Tt((y_i,k), (y_j, k)) \geq 0$ for all $i \in \N$, and all $j \geq (i+1)$.
Let $ \Vv_k^* := \Vv(x_0, y_0,k) $ and as function $\Vv$ is bounded from below for $k$, let the lower bound be $\Vv_k^{\dagger}$ for function $\Vv_k$. 
Following condition \eqref{eq:tbar_cond_3_inf_vis_counter} and via induction, we have $\Vv(x_0, y_{\ell},k) \leq  \Vv(x_0,y_0,k) - \ell\xi_k $ for all $\ell \in \N_{ \geq 1}$.
Thus there exists some $\ell \in \N$, such that $\Vv_i(x_0, y_\ell, k) \leq  \Vv_k^* - \ell\xi_k < \Vv_k^{\dagger} $ which is a contradiction.
\qed
\end{proof}
\section{Disjunctive C$^3$s for Infinite Visits}
\label{ap:ccc_disjunct_inf}
\begin{definition}
    \label{def:ccc_inf_vis_counter_disj}
    Consider a control system $\Sys = (\Xx, \Xx_0, U, f)$ and a set of states $\Xx_{INF}$ that must be visited infinitely often.
    Let $\Xx \setminus \Xx_{INF} = \underset{1 \leq i \leq p}{ \bigcup} \Xx_i$.
    Then, function $\Tt: (\Xx \times \R) \times (\Xx \times \R) \to \R$, and functions $\Vv_i: \Xx \times \Xx \times \R \to \R_{\geq 0}$ defined for all $1 \leq i \leq p$, that are bounded from below in $\Xx \times \Xx$ for every $j \in \R$, are a control closure certificate, if for all $x \in \Xx$, and $j \in \R$, there exists $u \in U$ such that:
    \begin{align}
        &\big( \Tt( (x,j), (f(x,u), k) ) \geq 0 \big), \label{eq:tbar_cond_1_inf_vis_counter_djunct}
    \end{align}
    where $k = j+1$ if $x \in \Xx_{INF}$, and $k = j$ otherwise.
    And for all $x, y \in \Xx$, $j \in \R$, $\ell \geq (j+1)$, and  $u \in U$, we have:
    \begin{align}
        & \big(  \Tt( (x,j), (f(x,u) ,k) \big) \geq 0 \big) \implies \nonumber \\
        & \Big( \big( \Tt( (f(x,u), k), (y, \ell)) \geq 0 \big) \implies \big( \Tt((x,j), (y, \ell)) \geq 0 \big) \Big), \label{eq:tbar_cond_2_inf_vis_counter_disjnuct}
    \end{align}
    where $k = j+1$ if $x \in \Xx_{INF}$, and $k = j$ otherwise.
    And for all $1 \leq i \leq p$, and all $j \in \R$, there exists $\xi_{i,j} \in \R_{ > 0}$ such that for all  $x_0 \in \Xx_0$, for all $z,z' \in \Xx_i$, we have:
    \begin{align}
            &\big( \Tt((x_0,0),(z,j)) \geq 0 \big) \wedge \big( \Tt((z,j),(z',j)) \geq 0 \big) \implies  \nonumber \\
         &\big( \Vv_i(x_0,z', j) \leq \Vv_i(x_0,z, j) - \xi_{i,j} \big).   \label{eq:tbar_cond_3_inf_vis_counter_disjnuct} 
    \end{align}
\end{definition}

\section{Results of Case Studies}
\label{ap:results}
The results of our case studies are shown below.

\subsection{Without DPA}
\label{ap:ss:no_dpa}
The monomials and corresponding coefficients for $\Tt(x,y)$ and $\Zz(x)$ for ensuring only finite visit are given in Table \ref{tab:finite_visits_no_dpa}.
Similarly, Table \ref{tab:infinite_visits_no_dpa} shows the coefficients for $\Tt(x,y)$ and $\Vv(x)$ for ensuring only infinite visits, and Table \ref{tab:fin_inf_visits_no_dpa} shows the coefficients for $\Tt(x,y)$, $\Zz(x)$ and $\Vv(x)$ for ensuring both finite and infinite visits.
The monomials are the same as the ones given in Table \ref{tab:finite_visits_no_dpa}. Observe that we drop the subscript indexing in Section \ref{sec:comp} when searching for C$^3$s.

\subsection{With DPA}
\label{ap:ss:with_dpa}
The monomials for $\Tt(x,y)$, $\Zz(x)$ and $\Vv(x)$ are the same as what are shown in \ref{tab:finite_visits_no_dpa}. The corresponding coefficients of $\Tt(x,y)$ and $\Zz(x)$ for ensuring only finite visits are given in Table \ref{tab:finite_visits_dpa}.
Similarly, Table \ref{tab:infinite_visits_dpa} shows the coefficients for $\Tt(x,y)$ and $\Vv(x)$ for ensuring only infinite visits, and Table \ref{tab:fin_inf_visits_dpa} shows the coefficients for $\Tt(x,y)$, $\Zz(x)$ and $\Vv(x)$ for ensuring both finite and infinite visits. Observe that when searching for C$^3$s, we drop the subscript indexing in Section \ref{sec:comp} and instead introduce the superscript indexing $(p,q)$ for automaton state pairs and $(q)$ for an automaton state.

\begin{table}[!t]
\centering
\begin{tabular}{|p{4em}|p{0.8\textwidth}|} 
  \hline
  $\mathbf{b}(x,y)^T$ &
  [$1$, $x_{1}$, $x_{2}$, $y_{1}$, $y_{2}$, $x_{1}^2$, $x_{1}x_{2}$, $x_{1}y_{1}$, $x_{1}y_{2}$, $x_{2}^2$, $x_{2}y_{1}$, $x_{2}y_{2}$, $y_{1}^2$, $y_{1}y_{2}$, $y_{2}^2$, $x_{1}^3$, $x_{1}^2x_{2}$, $x_{1}^2y_{1}$, $x_{1}^2y_{2}$, $x_{1}x_{2}^2$, $x_{1}x_{2}y_{1}$, $x_{1}x_{2}y_{2}$, $x_{1}y_{1}^2$, $x_{1}y_{1}y_{2}$, $x_{1}y_{2}^2$, $x_{2}^3$, $x_{2}^2y_{1}$, $x_{2}^2y_{2}$, $x_{2}y_{1}^2$, $x_{2}y_{1}y_{2}$, $x_{2}y_{2}^2$, $y_{1}^3$, $y_{1}^2y_{2}$, $y_{1}y_{2}^2$, $y_{2}^3$]
  \\
  \hline
  $\mathbf{c}_{\Tt}^T$ & [0.0, 0.0, 0.0, 0.0, 0.0, 201.843, -14.203, 0.0, 0.0, 172.842, 0.0, 0.0, $-234.516$, 10.871, -200.712, -73.892, 4.237, -0.809, -0.579, -90.436, 0.308, -2.24, 0.0, 0.0, 0.0, 3.133, 3.114, -1.086, 0.0, 0.0, 0.0, 68.172, -5.736, 89.093, 2.538]\\ 
  \hline
  $\mathbf{b}(x)^T$ & [$1$, $x_{1}$, $x_{2}$, $x_{1}^2$, $x_{1}x_{2}$, $x_{2}^2$, $x_{1}^3$, $x_{1}^2x_{2}$, $x_{1}x_{2}^2$, $x_{2}^3$]\\ 
  \hline
  $\mathbf{c}_{\Zz}^T$ & [11.74, 6.28, -0.239, 17.96, -1.598, 16.974, -5.414, -1.786, -6.245, 1.603]\\ 
  \hline
\end{tabular}
\caption{C$^3$ function parameters for verifying finite visits.}
\label{tab:finite_visits_no_dpa}
\end{table}

\begin{table}[!t]
\centering
\begin{tabular}{|p{2.5em}|p{0.8\textwidth}|} 
  \hline
  $\mathbf{c}_{\Tt}^T$ & [0.0, 0.0, 0.0, 0.0, 0.0, 272.602, -35.161, 0.0, 0.0, 227.626, 0.0, 0.0, -315.997, 31.929, -262.58, -123.719, 23.206, -0.545, -1.015, -152.684, $-0.459$, -2.093, 0.0, 0.0, 0.0, 21.691, 3.475, -1.021, 0.0, 0.0, 0.0, 116.451, -23.176, 152.285, -14.105]\\
  \hline
  $\mathbf{c}_{\Vv}^T$ & [15.978, 11.556, -3.123, 23.794, -3.589, 21.99, -13.244, 1.551, -10.816, 1.528]\\ 
  \hline
\end{tabular}
\caption{C$^3$ function parameters for verifying infinite visits.}
\label{tab:infinite_visits_no_dpa}
\end{table}

\begin{table}[!t]
\centering
\begin{tabular}{|p{2.5em}|p{0.8\textwidth}|} 
  \hline
  $\mathbf{c}_{\Tt}^T$ & [0.0, 0.0, 0.0, 0.0, 0.0, 229.094, -30.185, 0.0, 0.0, 189.194, 0.0, 0.0, $-265.751$, 27.07, -218.113, -103.806, 20.4, -0.539, -0.59, -127.156, $-0.333$, -1.633, 0.0, 0.0, 0.0, 13.119, 2.926, -1.218, 0.0, 0.0, 0.0, 98.212, -20.828, 125.42, -5.824]\\ 
  \hline
  $\mathbf{c}_{\Zz}^T$ & [16.511, 7.797, -1.113, 25.421, -2.668, 24.302, -8.324, -0.826, -10.269, 2.726]\\
  \hline
  $\mathbf{c}_{\Vv}^T$ & [16.781, 9.683, -2.84, 26.512, -2.773, 24.487, -12.942, 0.742, -11.804, 1.906]\\ 
  \hline
\end{tabular}
\caption{C$^3$ function parameters for verifying both finite and infinite visits.}
\label{tab:fin_inf_visits_no_dpa}
\end{table}

\begin{table}[!t]
\centering
\begin{tabular}{|p{2.5em}|p{0.8\textwidth}|}
  \hline
  $\mathbf{c}_{\Tt}^{(1,1)}$ & [89.996, 156.033, 23.408, 6.196, 35.292, 206.866, -92.511, 0.729, 7.894, 157.214, 2.485, 16.414, 59.013, 2.318, 34.885, 59.308, 63.135, 5.527, 40.86, 35.955, -1.789, -7.367, 11.174, 0.193, 4.979, 47.742, 5.148, 32.046, 24.523, 0.675, 10.611, 8.007, 19.08, 0.79, 14.229]$^T$\\
  \hline
  $\mathbf{c}_{\Tt}^{(1,2)}$ & [1.552, 123.315, -22.109, -126.67, 10.713, 82.713, -67.157, -84.825, 16.458, 20.893, 23.701, -15.942, 1.366, 17.894, 0.529, -5.495, 9.894, -54.431, 3.403, -10.329, 29.394, -11.163, 24.669, -20.222, -0.478, 3.184, -10.131, 0.557, 0.923, 7.075, -1.198, 19.815, -13.503, 5.491, 0.929]$^T$\\
  \hline
  $\mathbf{c}_{\Tt}^{(1,3)}$ & [3.335, 188.158, -17.719, -190.056, -0.164, 86.791, -115.61, -90.332, $-0.055$, 51.233, 11.201, -63.485, -11.844, 103.922, 23.013, 7.817, 21.025, -23.855, 20.72, -5.67, -34.793, -21.21, 10.15, -27.501, -3.647, 26.886, 13.486, -34.463, 37.815, 12.094, -5.394, -16.582, -22.413, $-13.575$, 12.579]$^T$\\
  \hline
  $\mathbf{c}_{\Tt}^{(2,1)}$ & [-12.234, 115.841, -4.916, -3.31, -15.88, 183.51, -123.22, -2.56, -9.662, 135.262, -0.831, -1.306, -24.46, -1.06, -15.683, 59.276, 49.078, 6.794, 35.893, 40.93, -4.926, -23.836, -13.878, -1.07, -8.67, 46.625, 6.495, 33.278, -1.849, -0.454, -2.587, -3.465, -8.938, -0.311, -5.876]$^T$\\
  \hline
  $\mathbf{c}_{\Tt}^{(2,2)}$ & [-13.553, 119.863, -13.744, -110.025, -0.048, 99.828, -81.658, -10.868, -0.597, 26.451, -3.375, 4.385, -89.661, 74.328, -29.29, -18.658, 26.761, -11.351, 4.339, -18.902, -1.833, -8.605, 15.162, 5.041, -2.01, 6.905, $-2.401$, 1.495, -3.287, 2.327, 2.364, -5.928, -22.049, 10.926, -7.457]$^T$\\
  \hline
  $\mathbf{c}_{\Tt}^{(2,3)}$ & [2.685, 187.494, -17.375, -190.056, -0.164, 66.91, -111.814, -24.541, -3.81, 35.786, -1.58, -29.852, -57.977, 114.047, 2.253, 1.267, 24.962, 13.926, 15.695, -3.657, -15.363, 3.35, -11.233, 7.093, -20.132, 21.533, 3.556, -22.9, 10.241, 10.327, 5.375, -26.883, -53.757, -18.798, -6.246]$^T$\\
  \hline
  $\mathbf{c}_{\Tt}^{(3,1)}$ & [-77.762, 190.056, 0.0, -2.885, -19.412, 176.558, -159.133, 0.0, -0.003, 132.287, 0.0, 0.001, -34.552, -1.258, -19.201, 71.198, 78.876, 3.897, 27.108, 56.705, -4.725, -27.116, 0.001, 0.0, 0.005, 46.509, 5.625, 33.037, 0.0, 0.0, -0.001, -4.542, -10.141, -0.479, -8.353]$^T$\\
  \hline
  $\mathbf{c}_{\Tt}^{(3,2)}$ & [-73.825, 190.056, 0.0, -126.741, -8.666, 97.947, -113.184, 0.0, -0.004, 32.656, 0.0, 0.0, -86.022, 80.219, -34.365, 0.935, 42.628, -10.351, 20.895, -12.227, -9.594, -22.264, 0.001, 0.0, 0.005, 14.533, 1.961, 1.784, 0.0, 0.0, 0.001, 6.871, -21.899, 14.603, -7.76]$^T$\\
  \hline
  $\mathbf{c}_{\Tt}^{(3,3)}$ & [0.0, 190.072, -0.003, -190.055, -0.17, 65.777, -125.166, 0.002, -0.066, -1.308, 0.0, 0.013, -77.103, 123.454, -28.145, 22.563, 25.987, -1.013, 14.806, 24.699, -12.717, -15.687, -0.017, -0.001, 0.071, 33.256, 1.111, -8.041, 0.004, 0.0, -0.014, -31.083, -51.608, -29.545, -10.623]$^T$\\
  \hline
  $\mathbf{c}_{\Zz}^{(2)}$ & [20.573, 48.074, -10.091, 76.913, -56.651, 25.141, 6.555, 9.809, -9.399, 7.194]$^T$\\ 
  \hline
\end{tabular}
\caption{C$^3$ function parameters for verifying finite visits.}
\label{tab:finite_visits_dpa}
\end{table}

\begin{table}[!t]
\centering
\begin{tabular}{|p{2.5em}|p{0.8\textwidth}|} 
  \hline
  $\mathbf{c}_{\Tt}^{(1,1)}$ & [4.003, 31.181, 1.953, -33.239, 3.176, 11.586, -6.606, -16.35, 0.289, 0.938, -0.736, 0.673, -12.434, 6.777, -4.226, -0.009, 1.547, -6.076, 0.191, -0.891, 0.074, -0.111, 0.909, -1.944, 0.269, 0.69, -0.529, 0.114, 1.006, -0.454, 0.033, -0.813, -0.372, 0.637, -0.4]$^T$\\
  \hline
  $\mathbf{c}_{\Tt}^{(1,2)}$ & [4.87, 92.446, -19.838, -97.432, 12.167, 55.392, -36.646, -47.169, 7.913, 8.029, 5.066, -4.216, -12.963, 17.417, -2.826, -5.031, 8.305, -32.301, 3.425, -6.717, 15.38, -8.886, 15.215, -15.646, 1.496, 1.15, -4.178, -0.618, 5.316, 0.111, 1.636, 11.45, -13.789, 5.155, -0.591]$^T$\\
  \hline
  $\mathbf{c}_{\Tt}^{(1,3)}$ & [4.268, 178.988, -17.4, -180.881, -0.103, 79.831, -109.797, -95.962, $-1.065$, 47.778, 12.425, -68.489, -0.669, 101.504, 30.335, 8.058, 24.667, -25.828, 22.828, -5.898, -35.285, -21.42, 10.175, -31.182, -4.052, 28.325, 13.017, -34.207, 38.998, 12.706, -7.075, -12.666, -24.888, -12.405, 13.251]$^T$\\
  \hline
  $\mathbf{c}_{\Tt}^{(2,1)}$ & [-226.174, 77.159, -16.395, -45.823, -59.825, 146.398, -85.587, -3.475, -0.216, 68.154, -2.652, 1.32, -59.563, 5.105, -25.231, 28.074, 39.422, 7.354, 32.053, 0.599, 0.801, 0.203, -0.54, 2.259, 2.893, 25.289, 7.864, 18.592, 2.342, 2.212, 1.793, 4.557, -5.607, 0.966, -2.171]$^T$\\
  \hline
  $\mathbf{c}_{\Tt}^{(2,2)}$ & [-14.076, 90.62, -15.9, -77.448, 2.612, 90.992, -72.405, -4.837, -2.274, 25.094, -3.294, 6.984, -87.175, 70.307, -31.08, -20.947, 32.604, -16.639, 5.107, -14.445, -0.281, -10.925, 23.924, 3.66, -0.543, 5.54, -2.758, 1.027, -6.209, 3.281, 2.601, -5.466, -24.448, 9.102, -6.146]$^T$\\
  \hline
  $\mathbf{c}_{\Tt}^{(2,3)}$ & [3.083, 178.709, -16.841, -180.881, -0.103, 59.335, -108.077, -25.296, -4.21, 32.428, -0.177, -32.889, -49.277, 111.14, 8.24, -3.627, 30.968, 16.814, 16.715, -5.707, -14.987, 7.89, -12.879, 6.458, -22.292, 21.497, 1.605, -21.302, 9.347, 11.299, 5.481, -23.642, -57.892, -19.685, -7.73]$^T$\\
  \hline
  $\mathbf{c}_{\Tt}^{(3,1)}$ & [-313.703, 180.881, 0.0, -45.207, -58.824, 144.116, -140.958, 0.0, -0.002, 88.061, 0.0, 0.0, -65.505, 6.775, -26.894, 48.282, 83.79, 4.941, 31.618, 22.813, -4.481, -19.126, 0.0, 0.0, 0.003, 37.711, 8.079, 30.526, 0.0, 0.0, -0.001, 3.121, -4.329, 1.311, -5.467]$^T$\\
  \hline
  $\mathbf{c}_{\Tt}^{(3,2)}$ & [-90.687, 180.881, 0.0, -93.648, -5.091, 90.39, -109.45, 0.0, -0.002, 36.364, 0.0, 0.0, -76.967, 73.359, -35.753, -1.254, 48.246, -9.865, 21.34, -12.164, -7.997, -21.605, 0.0, 0.0, 0.003, 11.883, 0.787, 1.927, 0.0, 0.0, 0.0, 7.537, -24.606, 14.864, -6.643]$^T$\\
  \hline
  $\mathbf{c}_{\Tt}^{(3,3)}$ & [0.0, 180.892, -0.002, -180.881, -0.107, 58.22, -121.39, 0.001, -0.042, -5.422, 0.0, 0.008, -69.373, 120.092, -23.107, 20.208, 31.314, -1.223, 15.112, 23.795, -12.432, -15.034, -0.01, 0.0, 0.045, 34.461, 1.043, -7.787, 0.003, 0.0, -0.009, -27.97, -56.099, -29.354, -11.655]$^T$\\ 
  \hline
  $\mathbf{c}_{\Vv}^{(1)}$ & [32.292, 28.092, -1.205, 19.631, -7.218, 3.942, -0.256, 1.281, -1.2, 0.397]$^T$\\
  \hline
  $\mathbf{c}_{\Vv}^{(2)}$ & [9.47, 26.13, -10.349, 55.307, -55.75, 22.926, -2.291, 15.594, -5.384, 4.494]$^T$\\
  \hline
\end{tabular}
\caption{C$^3$ function parameters for verifying infinite visits.}
\label{tab:infinite_visits_dpa}
\end{table}

\begin{table}[!t]
\centering
\begin{tabular}{|p{2.5em}|p{0.8\textwidth}|} 
  \hline
  $\mathbf{c}_{\Tt}^{(1,1)}$ & [3.414, 27.102, 1.706, -28.315, 3.031, 10.422, -5.12, -14.852, 0.416, 0.474, -0.89, 0.624, -10.525, 4.421, -3.067, 0.198, 1.175, -5.922, 0.231, -0.505, -0.101, -0.09, 0.485, -1.711, 0.239, 0.562, -0.517, 0.097, 0.937, -0.403, 0.036, -1.359, 0.148, 0.149, -0.284]$^T$\\
  \hline
  $\mathbf{c}_{\Tt}^{(1,2)}$ & [6.754, 75.948, -18.053, -82.059, 12.725, 47.298, -27.913, -38.963, 8.326, 4.414, 4.108, -3.92, -14.572, 9.845, -1.157, -1.921, 4.464, -28.659, 2.892, -4.173, 14.906, -6.508, 9.748, -11.127, 0.762, 0.658, -4.183, -0.475, 3.591, -0.128, 0.817, 12.928, -14.175, 4.201, -0.373]$^T$\\
  \hline
  $\mathbf{c}_{\Tt}^{(1,3)}$ & [4.198, 167.703, -16.268, -169.617, -0.154, 74.502, -103.355, -91.793, -1.226, 44.864, 11.517, -65.096, 1.154, 96.546, 29.217, 7.925, 23.646, $-24.738$, 22.115, -4.563, -33.502, -20.628, 9.752, -30.328, -4.278, 26.981, 12.268, -32.403, 36.902, 12.507, -7.234, -11.613, -23.43, -12.155, 12.907]$^T$\\
  \hline
  $\mathbf{c}_{\Tt}^{(2,1)}$ & [-211.997, 64.011, -19.751, -40.908, -58.591, 139.575, -72.534, -3.336, 0.266, 62.776, -2.658, 1.017, -56.673, 3.205, -24.087, 27.756, 34.731, 7.24, 31.214, 3.414, 0.726, -0.508, -0.525, 2.269, 2.793, 23.931, 7.508, 18.484, 2.059, 2.019, 1.611, 4.031, -5.305, 0.628, -2.148]$^T$\\
  \hline
  $\mathbf{c}_{\Tt}^{(2,2)}$ & [-13.987, 79.445, -17.21, -64.732, 4.975, 79.547, -60.244, -4.977, -0.502, 20.772, -2.91, 5.161, -75.761, 58.163, -26.638, -17.342, 27.125, -11.839, 1.95, -9.854, 0.859, -5.546, 15.991, 1.976, 0.014, 3.739, -2.42, 0.526, -5.789, 0.614, 1.276, -5.268, -17.449, 4.851, -3.89]$^T$\\
  \hline
  $\mathbf{c}_{\Tt}^{(2,3)}$ & [3.109, 167.608, -15.759, -169.616, -0.153, 54.595, -101.958, -24.286, -3.92, 30.413, 0.007, -31.493, -44.981, 105.51, 8.276, -4.445, 29.55, 17.02, 16.034, -4.344, -14.065, 8.291, -12.725, 5.948, -22.041, 19.981, 1.01, -19.836, 8.36, 11.395, 4.725, -22.012, -54.73, -19.885, -7.105]$^T$\\
  \hline
  $\mathbf{c}_{\Tt}^{(3,1)}$ & [-300.171, 169.616, 0.0, -40.285, -57.129, 136.489, -134.002, 0.0, -0.003, 83.14, 0.0, 0.001, -62.107, 4.874, -25.537, 46.671, 79.914, 4.814, 30.84, 23.719, -4.324, -19.021, 0.001, 0.0, 0.005, 36.276, 7.489, 29.327, 0.0, 0.0, -0.001, 2.712, -4.029, 0.927, -5.193]$^T$\\
  \hline
  $\mathbf{c}_{\Tt}^{(3,2)}$ & [-92.006, 169.616, 0.0, -79.077, -0.8, 83.739, -102.139, 0.0, -0.004, 33.795, 0.0, 0.0, -70.109, 61.042, -31.658, -1.448, 46.769, -9.072, 19.288, -11.087, -8.012, -18.518, 0.001, 0.0, 0.006, 11.686, 1.225, -0.341, 0.0, 0.0, 0.0, 6.984, -18.715, 9.873, -5.154]$^T$\\
  \hline
  $\mathbf{c}_{\Tt}^{(3,3)}$ & [0.0, 169.632, -0.003, -169.616, -0.159, 53.832, -115.197, 0.002, -0.063, -5.106, 0.0, 0.012, -64.46, 114.089, -21.604, 18.86, 30.138, -1.229, 14.409, 23.723, -11.924, -14.65, -0.015, -0.001, 0.068, 32.751, 0.999, -7.782, 0.004, 0.0, -0.013, -26.136, -53.189, -28.572, -10.726]$^T$\\
  \hline
  $\mathbf{c}_{\Zz}^{(2)}$ & [15.068, 27.023, -9.854, 54.129, -46.728, 20.994, -0.432, 10.042, -3.859, 3.144]$^T$\\ 
  \hline
  $\mathbf{c}_{\Vv}^{(1)}$ & [30.778, 23.608, -1.032, 17.916, -5.561, 3.146, 0.054, 1.007, -0.887, 0.301]$^T$\\
  \hline
  $\mathbf{c}_{\Vv}^{(2)}$ & [8.259, 21.928, -9.105, 49.06, -49.084, 20.147, -0.347, 11.934, -2.856, 2.899]$^T$\\
  \hline
\end{tabular}
\caption{C$^3$ function parameters for verifying both finite and infinite visits.}
\label{tab:fin_inf_visits_dpa}
\end{table}

\end{document}